\documentclass[11pt]{article}
\usepackage{amssymb,amsmath}%, fullpage}
\usepackage[margin=1in]{geometry}
\newtheorem{thm}{Theorem}[section]
\newtheorem{definition}[thm]{Definition}

\newtheorem{proposition}[thm]{Proposition}
\newtheorem{lemma}[thm]{Lemma}
\newtheorem{corollary}[thm]{Corollary}
\newtheorem{conjecture}[thm]{Conjecture}

\newtheorem{observation}[thm]{Observation}

\newenvironment{proof}[1][Proof:]
{\medskip \noindent {\bf #1} }{\hspace*{0pt}\hfill$\Box$ \bigskip}

\newcommand{\pr}{\mathrm{Pr}}
\newcommand{\expe}{\mathrm{E}}
\newcommand{\restrict}[1]{\upharpoonright_{#1}}
\newcommand{\ind}{\mathrm{ind}}
\newcommand{\ram}{\mathrm{Ram}}

\title{Testing Hereditary Properties of Ordered Graphs and Matrices} 
\author{Noga Alon\thanks{Sackler School of Mathematics
		and Blavatnik School of
		Computer Science, Tel Aviv University, Tel Aviv 69978, Israel.
		Email: {\tt nogaa@tau.ac.il}.  Research supported in part by a
		USA-Israeli
		BSF grant 2012/107, by an ISF grant 620/13 and
		by the Israeli I-Core program.} , 
	Omri Ben-Eliezer\thanks{Blavatnik School of
		Computer Science, Tel Aviv University, Tel Aviv 69978, Israel.
		Email: {\tt omribene@gmail.com}.} ,
	Eldar Fischer\thanks{Faculty of Computer Science, Israel Institute of Technology (Technion), Haifa, Isarel. Email: {\tt eldar@cs.technion.ac.il}.}}

\begin{document}

\begin{titlepage}
	\clearpage
	\maketitle	
	\thispagestyle{empty}

\begin{abstract}
We consider properties of edge-colored \emph{vertex-ordered graphs}, i.e., graphs with a totally ordered vertex set and a finite set of possible edge colors. We show that any hereditary property of such graphs is strongly testable, i.e., testable with a constant number of queries.
We also explain how the proof can be adapted to show that any hereditary property of $2$-dimensional matrices over a finite alphabet (where row and column order is not ignored) is strongly testable.
The first result generalizes the result of Alon and Shapira [FOCS'05; SICOMP'08], who showed that any hereditary graph property (without vertex order) is strongly testable. 
The second result answers and generalizes a conjecture of Alon, Fischer and Newman [SICOMP'07] concerning testing of matrix properties.
%To the best of our knowledge, these are the first testing results of this kind for the ordered versions of graphs and matrices.

The testability is proved by establishing a removal lemma for vertex-ordered graphs. It states that for any finite or infinite family $\mathcal{F}$ of forbidden vertex-ordered graphs, and any $\epsilon > 0$, there exist $\delta > 0$ and $k$ so that any vertex-ordered graph which is $\epsilon$-far from being $\mathcal{F}$-free contains at least $\delta n^{|F|}$ copies of some $F\in\mathcal{F}$ (with the correct vertex order) where $|F|\leq k$.
The proof bridges the gap between techniques related to the regularity lemma, used in the long chain of papers investigating graph testing, and string testing techniques. Along the way we develop a Ramsey-type lemma for $k$-partite graphs with ``undesirable'' edges, stating that one can find a Ramsey-type structure in such a graph, in which the density of the undesirable edges is not much higher than the density of those edges in the graph.

\end{abstract}
\end{titlepage}

%%%%%%%%%%%%%%%%%%%%
\section{Introduction}
%%%%%%%%%%%%%%%%%%%%
\emph{Property Testing} is dedicated to finding fast algorithms for decision problems of the following type:
Given a combinatorial structure $S$, distinguish quickly between the case where $S$
satisfies a property $\mathcal{P}$ and the case where $S$ is \emph{far} from satisfying the
property. Being far means that one needs to modify a significant fraction of the data in $S$
to make it satisfy $\mathcal{P}$.
Property Testing was first formally defined by Rubinfeld and Sudan
\cite{RubinfeldSudan1996}, and the investigation in the combinatorial context was initiated by Goldreich, Goldwasser and Ron \cite{GoldreichGoldwasserRon1998}.
This area has been very active over the last twenty years, see, e.g.\@ \cite{Goldreich2010} for various surveys on it. 

In this paper we focus on property testing of two-dimensional structures over a finite alphabet, or equivalently, two-variable functions with a fixed finite range.
Specifically, we consider graphs and matrices.
\emph{Graphs} are functions $G:\binom{V}{2} \to \{0,1\}$ where $V$ is the vertex set;
more generally \emph{edge-colored graphs} (with finite color set $\Sigma$) are functions $G:\binom{V}{2} \to \Sigma$. \emph{Matrices} over a finite alphabet $\Sigma$ (or \emph{images}) are functions $M:U \times V \to \Sigma$. In this paper we generally consider edge-colored graphs rather than standard graphs, as the added generality will prove useful later, so the term \emph{graph} usually refers to an edge-colored graph.

For a fixed finite set $\Sigma$, 
a \emph{property} of functions over $
\Sigma$ is simply a collection of functions whose range is $\Sigma$.
Specifically, an \emph{ordered graph property} is a collection of (edge-colored) graphs $G:\binom{V}{2} \to \Sigma$. An \emph{unordered graph property} is an ordered graph property that is also \emph{invariant under vertex permutations}: If $G \in \mathcal{P}$ and $\pi$ is any permutation on $V_G$, then the graph $G^{\pi}$, defined by $G^{\pi}(\pi(u) \pi(v)) = G(uv)$ for any $u \neq v \in V_G$, satisfies $G^{\pi} \in \mathcal{P}$.
Similarly, an \emph{(ordered) matrix property}, or an \emph{image property}, is a collection of functions $M:[m] \times [n] \to \Sigma$.
For simplicity, most definitions given below are only stated for graphs, but they carry over naturally to matrices.

A graph $G:\binom{[n]}{2} \to \Sigma$ is $\epsilon$-far from the property $\mathcal{P}$ if one needs to modify the value $G(ij)$ for at least $\epsilon \binom{n}{2}$ of the edges $ij$, where $ij$ denotes the (unordered) edge $\{i,j\} \in \binom{[n]}{2}$.
A \emph{tester} for the property $\mathcal{P}$ is a randomized algorithm that is given a parameter $\epsilon > 0$ and
query access to its input graph $G$.
The tester must distinguish, with error probability at most $1/3$, between the case where $G$ satisfies $\mathcal{P}$ and the case where $G$ is $\epsilon$-far from satisfying $\mathcal{P}$.
The tester is said to have \emph{one-sided error} if it always accepts inputs from $\mathcal{P}$, 
and rejects inputs that are $\epsilon$-far from $\mathcal{P}$ with probability at least $2/3$.
It is desirable to obtain testers that are efficient in terms of the \emph{query complexity} (i.e. the maximal possible number of queries made by the tester). A property $\mathcal{P}$ is \emph{strongly testable} if there is
a one-sided error tester for $\mathcal{P}$ whose query complexity is bounded by a function $Q(\mathcal{P}, \epsilon)$. In other words, the query complexity of the tester is independent of the size of the input. 
%$\mathcal{P}$ is \emph{easily testable} if $Q(\mathcal{P}, \epsilon)$ is polynomial in $\epsilon^{-1}$.

From now on, we generally assume (unless it is explicitly stated that we consider unordered graphs) that the vertex set $V$ of a graph $G$ has a total ordering (e.g. the natural one for $V = [n]$), which we denote by $<$.
The (induced) \emph{ordered subgraph} of the graph $G:
\binom{V}{2} \to \Sigma$ on $U \subseteq V$, where the elements of $U$ are $u_1 < \ldots < u_k$,
is the graph $H: [k] \to \Sigma$ which satisfies $H(ij) = G(u_i u_j)$ for any $i < j \in [k]$.
For a family $\mathcal{F}$ of ``forbidden'' graphs, the property $\mathcal{P}_{\mathcal{F}}$ of \emph{$\mathcal{F}$-freeness} consists of all graphs $G$ for which any ordered subgraph $H$ of $G$ satisfies $H \notin \mathcal{F}$.
Finally, a property $\mathcal{P}$ is \emph{hereditary} if it is closed under taking induced subgraphs. That is, for any $G \in \mathcal{P}$ and any ordered subgraph $H$ of $G$, it holds that $H \in \mathcal{P}$.
Note that a property $\mathcal{P}$ is hereditary if and only if $\mathcal{P} = \mathcal{P}_{\mathcal{F}}$ for some (finite or infinite) family $\mathcal{F}$ of graphs over $\Sigma$.

The analogous notions of ordered subgraphs, $\mathcal{F}$-freeness and hereditary properties for matrices are ``structure preserving''. Here, the \emph{ordered submatrix} of the matrix $M:[m] \times [n] \to \Sigma$ on $A \times B$, where the elements of $A$ and $B$ are $a_1 < \ldots < a_k$ and $b_1 < \ldots < b_l$, is the matrix $N:[k] \times [l] \to \Sigma$ defined by $N(i, j) = M(a_i, b_j)$ for any $i \in [k]$ and $j \in [l]$.

\subsection{Previous results on graphs and matrices}
\label{subsec:related_work}
Some of the most interesting results in property testing have been those that identify large families of properties that are efficiently testable, and those that show that large families of properties cannot be tested efficiently.

One of the most widely investigated questions in property testing has been that of characterizing the efficiently testable unordered graph properties. 
In the seminal paper of Goldreich, Goldwasser and Ron \cite{GoldreichGoldwasserRon1998} it was shown that all unordered graph properties that can be represented by a certain graph partitioning, including properties such as  $k$-colorability and having a large clique, are strongly testable. See also \cite{GoldreichTrevisan2003}.
Alon, Fischer, Krivelevich and Szegedy \cite{AlonFKS00} showed that the property of $\mathcal{F}$-freeness is strongly testable for any finite family $\mathcal{F}$ of forbidden unordered graphs (here the term \emph{unordered graphs} refers to the usual notion of graphs with no order on the vertices).
 Their main technical result, now known as the \emph{induced graph removal lemma}, is a generalization of the well-known \emph{graph removal lemma} \cite{AlonDuke1994, Szemeredi1978}. 
\begin{thm}[Induced graph removal lemma \cite{AlonFKS00}]
\label{thm:induced_graph_removal_lemma_AFKS}
For any finite family $\mathcal{F}$ of unordered graphs and $\epsilon > 0$ there exists $\delta = \delta(\mathcal{F}, \epsilon) > 0$, such that any graph $G$ which is $\epsilon$-far from $\mathcal{F}$-freeness contains at least $\delta n^q$ copies of some $F \in \mathcal{F}$ with $q$ vertices.
\end{thm}
The original proof of Theorem \ref{thm:induced_graph_removal_lemma_AFKS} uses
a strengthening of the celebrated Szemer\'edi graph regularity lemma \cite{Szemeredi1978}, known as the \emph{strong graph regularity lemma}. 

It is clear that having a removal lemma for a family $\mathcal{F}$ immediately implies that $\mathcal{F}$-freeness is strongly testable: A simple tester which picks a subgraph $H$ whose size depends only on $\mathcal{F}$ and $\epsilon$, and checks whether $H$ contains graphs from $\mathcal{F}$ or not, is a valid one-sided tester for $\mathcal{F}$-freeness. Hence, removal lemmas have a major role in property testing. They also have implications in different areas of mathematics, such as number theory and discrete geometry. For more details, see the survey of Conlon and Fox \cite{ConlonFox2013}.

By proving a variant of the induced graph removal lemma that also holds for infinite families, 
Alon and Shapira \cite{AlonShapira2008} generalized the results of \cite{AlonFKS00}.
The infinite variant is as follows.
\begin{thm}[Infinite graph removal lemma \cite{AlonShapira2008}]
	\label{thm:infinite_induced_graph_removal_lemma_AS}
	For any finite or infinite family $\mathcal{F}$ of unordered graphs and $\epsilon > 0$ there exist $\delta = \delta(\mathcal{F}, \epsilon) > 0$ and $q_0 = q_0(\mathcal{F}, \epsilon)$, such that any graph $G$ which is $\epsilon$-far from $\mathcal{F}$-freeness contains at least $\delta n^q$ copies of some $F \in \mathcal{F}$ on $q \leq q_0$ vertices.
\end{thm}
Theorem \ref{thm:infinite_induced_graph_removal_lemma_AS} directly implies that \emph{any} hereditary unordered graph property is strongly testable, exhibiting the remarkable strength of property testing.
\begin{thm}[Hereditary graph properties are strongly testable \cite{AlonShapira2008}]
	\label{thm:hereditary_unordered}
	Let $\Sigma$ be a finite set with $|\Sigma| \geq 2$. Any hereditary unordered graph property over $\Sigma$ is strongly testable.
\end{thm}
%Indeed, Theorem \ref{thm:hereditary_unordered} follows from Theorem \ref{thm:infinite_induced_graph_removal_lemma_AS} considering the following one-sided tester: given a graph $G$, query all edges of a random subgraph $H$ of $G$ whose size depends only on $\mathcal{P}$ and $\epsilon$. If $G$ is $\mathcal{F}$-free then $H$ must be $\mathcal{F}$-free, while if $G$ is 
%$\epsilon$-far from $\mathcal{F}$-freeness then the removal lemma implies that $H$ will contain a
%graph from $\mathcal{F}$ with high enough probability.
%
Alon, Fischer, Newman and Shapira later presented \cite{AlonFischerNewmanShapira2009} a complete combinatorial characterization of the graph properties that are testable (with two-sided error) using a constant number of queries, building on results from \cite{FischerNewman2005, GoldreichTrevisan2003}. Independently, Borgs, Chayes, Lov\'asz, S\'os, Szegedy and Vesztergombi
obtained an analytic characterization of such properties through the theory of graph limits \cite{BorgsCLSSV2006}. See also \cite{Lovasz2007, Lovasz2012}.

An efficient finite induced removal lemma for binary matrices with no row and column order was obtained by Alon, Fischer and Newman \cite{AlonFischerNewman2007}. In this case, $\delta^{-1}$ is polynomial in $\epsilon^{-1}$ (where $\epsilon, \delta$ play the same roles as in the above removal lemmas). It was later shown by Fischer and Rozenberg \cite{FischerRozenberg2007} that when the alphabet is bigger than binary, the dependence of $\delta^{-1}$ on $\epsilon^{-1}$ is super-polynomial in general, and in fact testing submatrix-freeness over a non-binary alphabet is at least as hard as testing triangle-freeness in graphs, for which the dependence is also known to be super-polynomial in general \cite{Alon2001}, see also \cite{AlonBenEliezer2016}.
Actually, the main tool in \cite{AlonFischerNewman2007} is an efficient conditional regularity lemma for ordered binary matrices, and it was conjectured there that
this regularity lemma can be used to obtain a removal lemma for ordered binary matrices.
\begin{conjecture}[Ordered binary matrix removal lemma \cite{AlonFischerNewman2007}]
	\label{conj:binary_matrix}
	For any finite family $\mathcal{F}$ of ordered binary matrices and any $\epsilon > 0$ there exists $\delta = \delta(\mathcal{F}, \epsilon)$ such that any $n \times n$ binary matrix which is $\epsilon$-far
	from $\mathcal{F}$-freeness contains at least $\delta n^{a+b}$ copies of some $a \times b$ matrix from $\mathcal{F}$.
\end{conjecture}
In contrast to the abundance of general testing results for two-dimensional structures with an inherent symmetry, such as unordered graphs and matrices, no similar results for \emph{ordered} two-dimensional structures (i.e. structures that do not have any underlying symmetry) have been established.
Even seemingly simple special cases, such as 
$F$-freeness for a single ordered graph $F$, or $M$-freeness for a single $2 \times 2$ ordered matrix $M$, are not known to be strongly testable in general \cite{AlonBenEliezer2016}.
A good survey on the role of symmetry in property testing is given by Sudan \cite{Sudan2010}, who suggests that the successful characterization of the strongly testable unordered graph properties is attributable to the underlying symmetry of these properties. See also \cite{GoldreichKaufman2011}.

Despite the lack of general results as above for the ordered case, 
property testing of multi-dimensional ordered structures has recently been an active area of research.
Notable examples of properties that were investigated in the setting of ordered matrices include monotonicity (see, e.g., \cite{ChakrabartySeshadhri2013, ChakrabartySeshadhri2014} for some of the recent works in the matrix setting), extensions of monotonicity such as $k$-monotonicity \cite{CanonneGGKW2016} and more generally poset properties \cite{FischerNewman2007},
visual and geometric properties of images, such as connectedness, convexity, being a half plane \cite{Raskhodnikova2003, BermanMR2016} and being a Lipschitz function \cite{AwasthiJMR2016, BlaisRY2014}, and local properties, such as consecutive pattern-freeness \cite{BenEliezerKormanReichman2017}.
Ordered graphs were less investigated in the context of property testing, but are the subject of many works in Combinatorics and other areas. See, e.g., a recent work on Ramsey-type questions in the ordered setting \cite{ConlonFLS2017}, in which it is shown that Ramsey numbers of simple ordered structures might differ significantly from their unordered counterparts.

Finally, we mention a relevant result on \emph{one-dimensional} structures. 
Alon, Krivelevich, Newman and Szegedy \cite{AlonKrivelevichNewmanSzegedy99} showed that 
regular languages are strongly testable. One can combine this result with the well-known Higman's lemma in order theory \cite{Higman1952} to show that any hereditary property of words (i.e.\@ one dimensional functions) over a finite alphabet is strongly testable.
\subsection{Our contributions}
We prove generalizations of Theorems \ref{thm:hereditary_unordered} and \ref{thm:infinite_induced_graph_removal_lemma_AS} to the ordered setting, as well as analogous results for matrices. The following result generalizes Theorem \ref{thm:hereditary_unordered}.
\begin{thm}[Hereditary properties of ordered graphs are strongly testable]
	\label{thm:hereditary_ordered_graphs}
	Fix a finite set $\Sigma$ with $|\Sigma| \geq 2$. Any hereditary ordered graph property over $\Sigma$ is strongly testable.
\end{thm}

To prove Theorem \ref{thm:hereditary_ordered_graphs}, we establish an \emph{order-preserving} induced graph removal lemma, which holds for finite and infinite families of ordered graphs. This is a generalization of Theorem \ref{thm:infinite_induced_graph_removal_lemma_AS}.

\begin{thm}[Infinite ordered graph removal lemma]
	\label{thm:infinite_removal_lemma}
	Fix a finite set $\Sigma$ with $|\Sigma| \geq 2$.
	For any (finite or infinite) family $\mathcal{F}$ of ordered graphs $F:\binom{[n_F]}{2} \to \Sigma$ 
	and any $\epsilon > 0$ there exist $q_0 = q_0(\mathcal{F}, \epsilon)$ and
	$\delta = \delta(\mathcal{F}, \epsilon) > 0$, such that any ordered graph $G:\binom{[n]}{2} \to \Sigma$ that is
	$\epsilon$-far from $\mathcal{F}$-freeness contains at least $\delta n^q$ induced copies of some graph $F \in \mathcal{F}$ on $q \leq q_0$ vertices.
\end{thm}

An analogue of Theorem \ref{thm:hereditary_ordered_graphs} for matrices is also proved. 
\begin{thm}[Hereditary properties of ordered matrices are strongly testable]
	\label{thm:hereditary_ordered_matrices}
	Fix a finite set $\Sigma$ with $|\Sigma| \geq 2$. Any hereditary (ordered) matrix property over $\Sigma$ is strongly testable.
\end{thm}
As in the case of ordered graphs, to prove Theorem \ref{thm:hereditary_ordered_matrices}
we establish the following ordered matrix removal lemma, which holds for finite and infinite families
of matrices, and settles a generalized form of Conjecture \ref{conj:binary_matrix}.
\begin{thm}[Infinite ordered matrix removal lemma]
	\label{thm:infinite_matrix_removal_lemma}
	Fix a finite set $\Sigma$ with $|\Sigma| \geq 2$.
	For any (finite or infinite) family $\mathcal{F}$ of ordered matrices over $\Sigma$
	and any $\epsilon > 0$ there exist $q_0 = q_0(\mathcal{F}, \epsilon) > 0$ and $\delta = \delta(\mathcal{F}, \epsilon) > 0$, such that any ordered matrix over $\Sigma$ that is
	$\epsilon$-far from $\mathcal{F}$-freeness contains at least $\delta n^{q + q'}$ copies of some $q \times q'$ matrix $F \in \mathcal{F}$, where $q, q' \leq q_0$.
\end{thm}
Actually, the proof of Theorem \ref{thm:infinite_matrix_removal_lemma} is almost identical to that of Theorem \ref{thm:infinite_removal_lemma}, so we only describe what modifications are needed to make the proof of Theorem \ref{thm:infinite_removal_lemma} also work here, for the case of square matrices. However, all proofs can be adapted to the non-square case as well. 
An outline for the proof of the graph case is given in Section \ref{sec:outline}, and all of the sections after it are dedicated to the full proof.
The needed modifications for the matrix case appear in Subsection \ref{subsec:adapt_matrices}.
%Similar statements hold for any vertex-colored ordered two-dimensional structure of the following type: The vertex colors are consecutive, i.e., if $x < y$ have the same color then there is no $x < z < y$ with a different color. Ordered graphs and matrices are special cases of such structures.

To the best of our knowledge, Theorems \ref{thm:hereditary_ordered_graphs} and \ref{thm:hereditary_ordered_matrices} are the first known testing results of this type for ordered two-dimensional structures, and Theorems \ref{thm:infinite_removal_lemma} and \ref{thm:infinite_matrix_removal_lemma} are the first known order-preserving removal lemmas for two-dimensional structures. 

It is interesting to note that some of the properties mentioned in Subsection \ref{subsec:related_work}, such as monotonicity, $k$-monotonicity, and forbidden-poset type properties in matrices, are hereditary (as all of them can be characterized by a finite set of forbidden submatrices), so Theorem \ref{thm:hereditary_ordered_matrices} gives a new proof that these properties, and many of their natural extensions, are strongly testable. 
Naturally, our general testers are much less efficient than the testers
specifically tailored for each of these properties (in terms of dependence
of the underlying constants on the parameters of the problem), but the
advantage of our result is its generality, that is, the fact that it
applies to any hereditary property. Thus, for example, for any fixed
ordered graph $H$ and any integer $k$, the property that an ordered graph
$G$ admits a $k$-edge coloring with no monochromatic (ordered) induced
copy of $H$ is strongly testable. As mentioned above Ramsey properties
of this type have been considered in the Combinatorics literature,
see \cite{ConlonFLS2017} and the references therein.  Another family
of examples includes properties of (integer) intervals on the line.
Any interval can be encoded by an edge connecting its two endpoints,
where the order on the vertices (the endpoints) is the usual order on
the real line.  A specific example of a hereditary property is that
the given set of intervals is closed under intersection. The forbidden
structure is a set of $4$ vertices $i<j<k<l$ where $ik$ and $jl$ are edges
(representing intervals) whereas $jk$ is a non-edge.

Finally, there are various examples of unordered hereditary graph properties that have simple %and natural 
representations using a small \emph{finite} forbidden family of \emph{ordered} subgraphs, while in the unordered representation, the forbidden family is infinite. 
Some examples of such properties are bipartiteness, being a chordal graph, and being a strongly chordal graph  \cite{Damaschke1990, BrandstadtLeSpinrad1999}.  
For such properties, when the input graph is supplied with the ``right'' ordering of the vertices, one can derive the strong testability using the version of Theorem \ref{thm:infinite_removal_lemma} for \emph{finite} families of forbidden ordered subgraphs -- see Theorem \ref{thm:finite_removal_lemma} below -- instead of using the infinite unordered version, Theorem \ref{thm:infinite_induced_graph_removal_lemma_AS}.

\subsection{Discussion and open questions}
\label{subsec:open_questions}
Several possible directions for future research follow from our work.

\subsubsection*{Dependence between the parameters of the ordered removal lemmas}
Our proofs rely heavily on strong variants of the graph regularity lemma.
Regularity-based proofs generally have a notoriously bad dependence between the parameters of the problem. In the notation of Theorem \ref{thm:infinite_removal_lemma}, for a fixed finite family $\mathcal{F}$ of forbidden ordered subgraphs, $\delta^{-1}$ is generally very large in terms of $\epsilon^{-1}$, meaning that the number of queries required for the corresponding tester for such properties is very large in terms of $\epsilon^{-1}$.
Indeed, the original Szemer\'edi regularity lemma imposes a tower-type dependence between these parameters \cite{Gowers1997, LovaszFox2014, MoshkovitzShapira2016}, while the variant we use is at least as strong (and at least as expensive) as the strong regularity lemma \cite{AlonFKS00}, which is known to have a wowzer (tower of towers) type dependence between its parameters \cite{ConlonFox2012, KalyanasundaramShapira2012}. 
Note that for infinite families $\mathcal{F}$ the dependence between the parameters may be arbitrarily bad \cite{AlonShapira2008-2}. 

In a breakthrough result of Fox \cite{Fox2011}, the first known proof for the (unordered) graph removal lemma that does not use the regularity lemma is given. However, the dependence between the parameters there is still of a tower type. In any case, it will be interesting to try to obtain a proof for the ordered case, that does not go through the strong regularity needed in our proof.

\subsubsection*{Better dependence for specific properties}
As discussed in Subsection \ref{subsec:related_work}, for ordered binary matrices there is an efficient conditional regularity lemma \cite{AlonFischerNewman2007}, in which the dependence of $\delta^{-1}$ on $\epsilon^{-1}$ is polynomial. It will be interesting to try to combine the ideas from our proof with this binary matrix regularity lemma, to obtain a removal lemma for finite families of ordered binary matrices with better dependence between the parameters. Ideally, one hopes for a removal lemma with polynomial dependence, but even obtaining such a lemma with, say, exponential dependence will be interesting.

More generally, it will be interesting to find large families of hereditary ordered graph or matrix properties that have more efficient testers than those obtained from our work. See, e.g., \cite{GishbolinerShapira2016} for recent results of this type for unordered graph properties.

\subsubsection*{Characterization of strongly testable ordered properties}
For unordered graphs, Alon and Shapira \cite{AlonShapira2008} showed that a property is strongly testable using an \emph{oblivious} one-sided tester, which is a tester whose behavior is independent of the size of the input, \emph{if and only if} the property is (almost) hereditary.
It will be interesting to obtain similar characterizations in the ordered case.

More generally, in the ordered case there are other general types of properties that may be of interest. Ben-Eliezer, Korman and Reichman \cite{BenEliezerKormanReichman2017} recently raised the question of characterizing the efficiently testable \emph{local properties}, i.e., properties that are characterized by a collection of forbidden local substructures.
It will also be interesting to identify and investigate large classes of \emph{visual} (or \emph{geometric}) properties.
Due to the lack of symmetry, obtaining a complete characterization of the efficiently testable properties of ordered graphs and matrices seems to be very difficult. In fact, considering that all properties whatsoever can be formulated as properties of ordered structures (e.g. strings), any characterization here will have to define and refer to some ``graphness'' of our setting, even that we do not allow the graph symmetries.

\subsubsection*{Generalization to ordered hypergraphs and hypermatrices}
It will be interesting to obtain similar removal lemmas (and consequently, testing results) for the high-dimensional analogues of ordered graphs and matrices, namely ordered $k$-uniform hypergraphs and $k$-dimensional hypermatrices. Such results were proved for unordered hypergraphs \cite{RodlSkokan2004, NagleRodlSchacht2006, Tao2006}.

\subsubsection*{Analytic analogues via graph limits}
The theory of graph limits has provided a powerful approach for problems of this type in the unordered case \cite{BorgsCLSSV2006, Lovasz2007, Lovasz2012}. It will be interesting to define and investigate a limit object for ordered graphs; this may also help with the characterization question above.
%%%%%%%%%%%%%%%%%%%%

\section{Paper outline}
\label{sec:outline}

A proof of a graph removal lemma typically goes along the following lines: First, the vertex set of the graph is partitioned into a ``constant'' (not depending on the input graph size itself) number of parts, and a corresponding {\em regularity scheme} is found. The regularity scheme essentially allows that instead of considering the original graph, one can consider a very simplified picture of a constant size structure approximately representing the graph. On one hand, the structure has to approximate the original graph in the sense that we can ``clean'' the graph, changing only a small fraction of the edges, so that the new graph will not contain anything not already ``predicted'' by the representing structure. On the other hand, the structure has to be ``truthful'', in the sense that everything predicted by it in fact already exists in the graph.

In the simplest case, just a regular partition given by the original Szemer\'edi Lemma would suffice. More complex cases, like \cite{AlonFKS00,AlonShapira2008}, require a more elaborate regularity scheme. In our case, Section \ref{sec:regscheme} provides a regularity scheme that addresses both edge configuration and vertex order, combining a graph regularity scheme with a scheme for strings.

Given a regularity scheme, we have to provide the graph cleaning procedure, as well as prove that if the cleaned graph still contains a forbidden subgraph, then the original graph already contains a structure containing many such graphs (this will consist of some vertex sets referenced in the regularity scheme). In Section \ref{sec:proof_main_result} we show how to use the scheme to prove the removal lemma and the testability theorem for the case of a finite family $\mathcal{F}$ of forbidden subgraphs, while in Section \ref{sec:proofinfinite} we show how to extend it for the case of a possibly infinite family $\mathcal{F}$. The latter section also contains a formal definition of what it means for the regularity scheme to predict the existence of a forbidden subgraph, while for the finite case it is enough to keep it implicit.

To extract the regularity scheme we need two technical aids. One of which, in Subsection \ref{sub:rounding}, is just a rounding lemma that allows us to properly use integer quantities to approximate real ones. While in many works the question of dealing with issues related to the divisibility of the number of vertices is just hand-waved away, the situation here is complex enough to merit a formal explanation of how rounding works.

In Subsection \ref{sub:ramsey} we develop a Ramsey-type theorem that we believe to be interesting in its own right. The use of Ramsey-type theorems is prevalent in nearly all works dealing with regularity schemes, as a way to allow us to concentrate only on ``well-behaved'' structures in the scheme when we are about to clean the graph. Because of the extra complication of dealing with vertex-ordered graphs, we cannot just find Ramsey-type instances separately in different parts of the regularity scheme. Instead, we need to find the well-behaved structure ``all at once'', and furthermore assure that we avoid enough of the ``undesirable'' parts where the regularity scheme does not reflect the graph. The fraction of undesirable features, while not large, must not depend on any parameters apart from the original distance parameter $\epsilon$ (and in particular must not depend on the size of the regularity scheme), which requires us to develop the new Ramsey-type theorem.

Roughly speaking, the theorem states the following: If we have a $k$-partite edge-colored graph with sufficiently many vertices in each part, then we can find a subgraph where the edges between every two parts are of a single color (determined by the identity of the two parts). However, we do it in a way that satisfies another requirement: If additionally the original graph is supplied with a set of ``undesirable'' edges comprising an $\alpha$ fraction of the total number of edges, then the subgraph we find will include not more than an $(1+\eta)\alpha$ fraction of the undesirable edges, for an $\eta$ as small as we would like (in our application $\eta=1$ will suffice).

\subsection{Finding a regularity scheme}
\label{subsec:finding_scheme}
To prove the removal lemma we need a {\em regularity scheme}, that is a sequence of vertex sets whose ``interaction'' with the graph edges, and in our case also the graph vertex order, allows us to carry a cleaning procedure using combinatorial lemmas.

Historically, in the case of properties like triangle-freeness in ordinary graphs, a regular equipartition served well enough as a regularity scheme. One needs then to just remove all edges that are outside the reach of regularity, such as edges between the sets that do not form regular pairs. When moving on to more general properties of graphs, this is not enough. We need a robust partition (see \cite{FischerNewman2005}) instead of just a regular one, and then we can find a subset in each of the partition sets so that these ``representative'' sets will all form regular pairs. This allows us to decide what to do with problem pairs, e.g. whether they should become complete bipartite graphs or become edgeless (we also need to decide what happens {\em inside} each partition set, but we skip this issue in the sketch).

For vertex ordered graphs, a single robust partition will not do. The reason is that even if we find induced subgraphs using sets of this partition, there will be no guarantees about the vertex order in these subgraphs. The reason is that the sets of the robust partition could interact in complex ways with regards to the vertex order. Ideally we would like every pair of vertex sets to appear in one of the following two possible ways:
Either one is completely before the other, or the two are completely ``interwoven''.

To interact with the vertex order, we consider the robust partition along with a secondary interval partition. If we consider what happens between two intervals, then all vertices in one of the intervals will be before all vertices in the other one. This suggests that further dividing a robust partition according to intervals is a good idea. However, we also need that inside each interval, the relevant robust partition sets will be completely interwoven. In more explicit terms, we will consider what happens when we intersect them with intervals of a {\em refinement} of the original interval partition. If these intersections all have the ``correct'' sizes in relation of the original interval (i.e., a set that intersects an interval also intersects all relevant sub-intervals with sufficient vertex count), then we will have the ``every possible order'' guarantee.

Section \ref{sec:regscheme} is dedicated to the formulation and existence proof of a regularity scheme suitable for ordered graphs.
In Subsection \ref{sub:approx} we  present the concept of \emph{approximating partitions}, showing several useful properties of them. Importantly, the notion of a robust partition is somewhat preserved when moving to a partition approximating it.

In Subsection \ref{sub:lcrscheme} we develop the lemma that gives us the required scheme. Roughly speaking, it follows the following steps.

\begin{itemize}
\item We find a base partition $P$ of the graph $G$, robust enough with regards to the graph edge colors, so as to ensure that it remains robust even after refining it to make it fit into a secondary interval partition.
\item We consider an interval partition $J$ of the vertex set $V$ of $G$, that is robust with respect to $P$. That is, if we partition each interval of $J$ into a number of smaller intervals (thus obtaining a refinement $J'$), most of the smaller intervals will contain about the same ratio of members of each set of $P$ as their corresponding bigger intervals.
\item Now we consider what happens if we construct a partition resulting from taking the intersections of the members of $P$ with members of $J'$. In an ideal world, if a set of $P$ intersects an interval of $J$, then it would intersect ``nicely'' also the intervals of $J'$ that are contained in that interval. However, this is only mostly true. Also, this ``partition by intersections'' will usually not be an equipartition.

\item We now modify a bit both $P$ and $J$, to get $Q$ and $I$ that behave like the ideal picture, and are close enough to $P$ and $J$. Essentially we move vertices around in $P$ to make the intersections with the intervals in $J'$ have about the same size inside each interval of $J$. We also modify the intersection set sizes (which also affects $J$ a little) so they will all be near multiples of a common value (on the order of $n$). This is so we can divide them further into an equipartition that refines both the robust graph partition and the interval partition. The rounding Lemma \ref{lem:round} helps us here.
\end{itemize}

The above process generates the following scheme. 
$Q$ is the modified base equipartition, and its size (i.e. number of parts) is denoted by $k$. $I$ is the modified ``bigger intervals'' equipartition, and its size is denoted by $m$. We are allowed to require in advance that $m$ will be large enough (that is, to have $m$ bigger than a predetermined constant $m_0$).
There is an equipartition $Q'$ of size $mt$ which refines both $Q$ and $I$. That is, each part of $Q'$ is fully contained in a part of $Q$ and a part of $I$, and so each part of $Q$ contains exactly $t$ parts of $Q'$. Moreover, there is the ``smaller intervals'' equipartition $I'$ which refines $I$, and has size $m b$ where $b = r(m,t)$ for a two-variable function $r$ that we are allowed to choose in advance ($r$ is eventually chosen according to the Ramsey-type arguments needed in the proof). Each part of $I$ contains exactly $b$ parts of $I'$.
Finally, there is a ``perfect" equipartition $Q''$ which refines $Q'$ and $I'$ and has size $mbt$,
such that inside any bigger interval from $I$, the intersection of each part of $Q'$ with each smaller interval from $I'$ consists of exactly one part of $Q''$.
Additionally, $Q'$ can be taken to be very robust, where we are allowed to choose the robustness parameters in advance.

We are guaranteed that the numbers $m$ and $t$ are bounded in terms of the above function $r$, the robustness parameters, and $m_0$ for which we required that $m \geq m_0$. These bounds do not depend on the size of the input graph. See Lemma \ref{lem:mainpartscheme} for more details.

\subsection{Proving a finite removal lemma}
\label{subsec:main_proof_sketch}
Consider an ordered colored graph $G:\binom{[n]}{2} \to \Sigma$, and consider a regularity scheme consisting of equipartitions $Q, I, Q', I', Q''$ for $G$ as described above. %For now we do not assume that any of the guarantees in the end of Subsection \ref{subsec:finding_scheme} hold. That is, we do not assume that $Q'$ is robust, nor that $m$ or $t$ are bounded. 

We start by observing that if $Q''$ is robust enough, then there is a tuple $W$ of ``representatives'' for $Q''$, satisfying the following conditions.
\begin{itemize}
\item For each part of $Q''$ there is exactly one representative, which is a subset of this part.
\item Each representative is not too small: it is of order $n$ (where the constants here may depend on all other parameters discussed above, but not on the input size $n$).
\item All pairs of representatives are very regular (in the standard Szemer\'edi regularity sense).
\item The densities of the colors from $\Sigma$ between pairs of representatives are usually similar to the densities of those colors between the pairs of parts of $Q''$ containing them. Here the \emph{density} of a color $\sigma \in \Sigma$ between vertex sets $A$ and $B$ is the fraction of $\sigma$-colored edges in $A \times B$.
\end{itemize}
Actually, the idea of using representatives, as presented above, was first developed in \cite{AlonFKS00}. 
Note that each part of $Q'$ contains exactly $b$ representatives (since it contains $b$ parts from $Q''$) and each small interval of $I'$ contains exactly $t$ representatives.

Now if $Q'$ is robust enough then the above representatives for $Q''$ also represent $Q'$ in the following sense: 
Densities of colors between pairs of representatives are usually similar to the densities of those colors between the pairs of parts of $Q'$ containing them.

Consider a colored graph $H$ whose vertices are the small intervals of $I'$,
where the ``color'' of the edge between two vertices (i.e.\@ small intervals) is the $t \times t$ ``density matrix'' described as follows: For any pair of representatives, one from each small interval, there is an entry in the density matrix. This entry is the set of all colors from $\Sigma$ that are dense enough between these two representatives, i.e., all colors whose density between these representatives is above some threshold.

An edge between two vertices of $H$ is considered \emph{undesirable} if the density matrix between these intervals differs significantly from a density matrix of the large intervals from $I$ containing them. 
If $Q'$ is robust enough, then most density matrices for pairs of small intervals are similar to the density matrices of the pairs of large intervals containing them. 
Therefore, the number of undesirables in $H$ is small in this case. 

Consider now $H$ as an $m$-partite graph, where each part consists of all of the vertices (small intervals) of $H$ that are contained in a certain large interval from $I$.
We apply the undesirability-preserving Ramsey on $H$, and then a standard multicolored Ramsey within each part, to obtain an induced subgraph $D$ of $H$ with the following properties.
\begin{itemize}
	\item $D$ has exactly $d_{\mathcal{F}}$ vertices (small intervals) inside each part of $H$, where $d_{\mathcal{F}}$ is the maximum number of vertices in a graph from the forbidden family $\mathcal{F}$.
	\item For any pair of parts of $H$, all $D$-edges between these parts have the same ``color'', i.e. the same density matrix.
	\item For any part of $H$, all $D$-edges inside this part have the same ``color''.
	\item The fraction of undesirables among the edges of $D$ is small.
\end{itemize}

Finally we wish to ``clean'' the original graph $G$ as dictated by $D$.
For any pair $Q'_1, Q'_2$ of (not necessarily distinct) parts from $Q'$, let $I_1, I_2$ be the large intervals from $I'$ containing them, and consider the density matrix that is common to all $D$-edges between $I_1$ and $I_2$.
In this matrix there is an entry dedicated to the pair $Q'_1, Q'_2$, which we refer to as the set of colors from $\Sigma$ that are ``allowed'' for this pair.
The cleaning of $G$ is done as follows: For every $u \in Q'_1$ and $v \in Q'_2$, if the original color of $uv$ in $G$ is allowed, then we do not recolor $uv$. Otherwise, we change the color of $uv$ to one of the allowed colors.

It can be shown that if $D$ does not contain many undesirables, then the cleaning does not change the colors of many edges in $G$. Therefore, if initially $G$ is $\epsilon$-far from $\mathcal{F}$-freeness, then there exists an induced copy of a graph $F \in \mathcal{F}$ in $G$ with $l \leq d_{\mathcal{F}}$ vertices after the cleaning. 
Considering our cleaning method, it can then be shown that there exist representatives $R_1, \ldots, R_l$ with the following properties. For any $i$, all vertices of $R_i$ come before all vertices of $R_{i+1}$ in the ordering of the vertices, and for any $ i < j$, the color of $F(ij)$ has high density in $R_i \times R_j$.
Recalling that all pairs of representatives are very regular, a well-known lemma implies that the representatives $R_1, \ldots, R_l$ span many copies of $F$, as desired.

\subsection{From finite to infinite removal lemma}
After the finite removal lemma is established, adapting the proof to the infinite case is surprisingly not difficult. 
The only problem of the finite proof is that we required $D$ to have exactly $d_{\mathcal{F}}$ vertices in each large interval, where $d_{\mathcal{F}}$ is the maximal number of vertices of a graph in $\mathcal{F}$. This requirement does not make sense when $\mathcal{F}$ is infinite. 
Instead we show that there is a function $d_{\mathcal{F}}(m,t)$ that ``plays the role'' of $d_{\mathcal{F}}$ in the infinite case. 

$d_{\mathcal{F}}(m,t)$ is roughly defined as follows: We consider the (finite) collection $\mathcal{C}(m,t)$ of all colored graphs with loops that have exactly $m$ vertices, where the set of possible colors is the same as that of $H$ (so the number of possible colors depends only on $|\Sigma|$ and $t$). We take $d_{\mathcal{F}}(m,t)$ to be the smallest number that guarantees the following. If a graph $C \in \mathcal{C}(m,t)$ exhibits (in some sense) a graph from $\mathcal{F}$, then $C$ also exhibits a graph from $\mathcal{F}$ with no more than $d_{\mathcal{F}}(m,t)$ vertices.

The rest of the proof follows as in the finite case, replacing any occurrence of $d_\mathcal{F}$ in the proof with $d_{\mathcal{F}}(m,t)$. 
Here, if $G$ contains a copy of a graph from $\mathcal{F}$ after the cleaning, then there is a set of no more than $d_{\mathcal{F}}(m, t)$  different representatives that are very regular in pairs and have the ``right'' densities with respect to some $F \in \mathcal{F}$ with at most $d_{\mathcal{F}}(m,t)$ vertices, so we are done as in the finite case.

\paragraph{From ordered graphs to ordered matrices}
To prove Theorem \ref{thm:infinite_matrix_removal_lemma} for square matrices, we reduce the problem to a graph setting. Suppose that $M:U \times V \to \Sigma$ is a matrix,
and add an additional color $\sigma_0$ to $\Sigma$. All edges between $U$ and $V$ will have the original colors from $\Sigma$, and edges inside $U$ and inside $V$ will have the new color $\sigma_0$.
Note that we are not allowed to change colors to or from the color $\sigma_0$, as it actually signals ``no edge''.
The proof now follows from the proof for graphs: We can ask the partition $I$ into large intervals
to ``respect the middle'', so all parts of $I$ are either fully contained in $U$ or in $V$. Moreover, 
colors of edges inside $U$ or inside $V$ are not modified during the cleaning step, and edges between $U$ and $V$ are not recolored to $\sigma_0$, since this color does not appear at all between the relevant representatives (and in particular, does not appear with high density).

To adapt the proof of Theorem \ref{thm:infinite_matrix_removal_lemma} for non-square matrices, we need the divisibility condition to be slightly different than respecting the middle. In the case that $m = o(n)$, we need to construct two separate ``large intervals'' equipartitions, one for the rows and one for the columns, instead of one such equipartition $I$ as in the graph case. The rest of the proof does not change.

%%%%%%%%%%%%%%%%%%%%
\section{Preliminaries and definitions}
%%%%%%%%%%%%%%%%%%%%
In general, we may and will assume whenever needed throughout the paper that $n$ is large enough with respect to all relevant parameters.
We generally denote ``small'' parameters and functions (whose values are always positive but can be arbitrarily close to zero) by small Greek letters\footnote{The only exception is $\lambda$, which will denote general real numbers, and $\ell$ which will denote their rounding},  and ``large'' parameters and functions (whose values are always finite natural numbers but can be arbitrarily large) by Latin letters. We assume that all parameters in all statements of the lemmas are monotone in the ``natural'' direction, as in the following examples: $T(\alpha, b) \leq T(\alpha', b')$
for $\alpha \geq \alpha', b \leq b'$, and $\gamma(c, \delta) \leq \gamma(c', \delta')$ for $c \geq c', \delta \leq \delta'$.
We also assume that all ``small'' parameters are smaller than one, and all ``large'' parameters are larger than one.

We also assume that all functions are ``bounded by their parameters'', for example $\gamma(\alpha, k) \leq \alpha$ and $T(\alpha, k) \geq k$.
These definitions extend naturally to any set of parameters, and are easily seen to be without loss of generality as long as we do not try to optimize bounds.

\subsubsection*{Colored graphs and charts}
A \emph{$\Sigma$-colored graph} $G = (V, c_G)$ is defined by a totally ordered set of vertices $V$ and a function $c_G : \binom{V}{2} \to \Sigma$. That is, $G$ is a complete ordered graph whose edges are colored by elements of $\Sigma$.
The standard notion of an (ordered) graph is equivalent to a $\{0,1\}$-colored graph.
A \emph{$\Sigma$-colored graph with loops} $G' = (V, c_{G'})$ is defined by a totally ordered set $V$
and a function $c_{G'} : \binom{V}{2} \cup V \to \Sigma$. We identify the notation $c_{G'}(vv)$ with $c_{G'}(v)$ for any $v \in V$.

With a slight abuse of notation, we denote by $U_1 \times U_2 = \{\{u_1, u_2\} : u_1 \in U_1, u_2 \in U_2\}$ the set of edges between two disjoint vertex sets $U_1$ and $U_2$.
A \emph{$(k, \Sigma)$-chart} $C=(V_1,\ldots,V_k,c_C)$ is defined by $k$ disjoint vertex sets $V_1,\ldots,V_k$ and a function 
$c_C:E_C \to \Sigma$, where $E_C = \bigcup_{1 \leq i < j \leq k} U_i \times U_j$. In other words, it is an edge-colored complete $k$-partite graph. 
For $C$ and $G$ as above, we say that $C$ is a \emph{partition} of $G$ if $V = \bigcup_{i=1}^k V_i$ and $c_G(e) = c_C(e)$ for any edge $e \in E_C$.
Moreover, $C$ is \emph{equitable} if $||V_i| - |V_j|| \leq 1$ for any $1 \leq i,j \leq k$; an equitable partition is sometimes called an \emph{equipartition}. 
The \emph{size} $|C|$ of the partition $C$ is the number of parts in it.
For a partition $C$ as above,
a $(k', \Sigma)$-chart $C'$ which is also a partition of $G$ is said to be a $G$-\emph{refinement} %(or sometimes simply a \emph{refinement}) 
of $C$ if we can write
$C' = (V_{1 1}, \ldots, V_{1 j_1}, \ldots, V_{k 1}, \ldots, V_{k j_k}, c_{C'})$ where $V_i = \bigcup_{l = 1}^{j_i} V_{i l}$. Note that $c_G(e) = c_C(e) = c_{C'}(e)$ for any edge $e \in E_C$. 
We will sometimes omit the coloring from the description of a partition when it is clear from the context (as the coloring is determined by the partition of the vertices and the coloring of the graph).
%In the rest of this section, $C,C'$ will always denote equipartitions of $G$ where $C'$ is a $G$-refinement of $C$.

For two disjoint sets of vertices $U, W$ and a coloring 
$c:U \times W \to \Sigma$,
we say that the \emph{density} of $\sigma \in \Sigma$ in $(U, W, c)$ is $d_\sigma(U, W, c) = |(U \times W) \cap c^{-1}(\sigma)| / |U||W|$. 
the squared density is denoted by $d_\sigma^2(U, W, c)$.
The \emph{index} of $(U, W, c)$ is
\[
\ind(U, W, c) = \sum_{\sigma \in \Sigma} d_\sigma^2(U, W, c).
\]
Note that $0 \leq \ind(U,V,c) \leq 1$ always holds. When the coloring $c$ is clear from context, we will usually simply write $d_\sigma(U,V)$ for density and $\ind(U,V)$ for index.

For a chart $C$ as above we define the \emph{index} of $C$ as
\[
\ind(C) = \sum_{1 \leq i < i' \leq k} \frac{|V_i| |V_{i'}|}{\binom{|V|}{2}}\ind(V_{i}, V_{i'}, c\restrict{V_{i} \times V_{i'}})
\]
where $V = \bigcup_{i=1}^{k} V_i$. %  as above.
By the Cauchy-Schwarz inequality, for any two partitions $C$, $C'$ of $G$ where $C'$ is a $G$-refinement of $C$ we have
\begin{equation}
\label{eq:index_refinement}
0 \leq \ind(C) \leq \ind(C') \leq 1.
\end{equation}
For a function $f: \mathbb{N} \to \mathbb{N}$ and a constant $\gamma > 0$, we say that an equipartition $C$ of size $k$ is 
\emph{$(f, \gamma)$-robust} if there exists no refining equipartition $C'$ of $C$ of size at most $f(k)$ for which $\ind(C') > \ind(C) + \gamma$.
The following observation 
states that for any colored graph $G$ and any equipartition $C$ of $G$, there exists an  $(f, \gamma)$-robust equipartition $C'$ refining $C$.
The first explicit definition of robustness was given in \cite{FischerNewman2005}.
\begin{observation}[Robust partitioning of colored graphs \cite{FischerNewman2005}]
\label{obs:robust_partition_colored_graphs}
For any integer $k > 0$, function $f:\mathbb{N} \to \mathbb{N}$ and real $\gamma > 0$ there exists 
$T = T_{\ref{obs:robust_partition_colored_graphs}}(k, f, \gamma)$ such that for any equipartition $C$ of a colored graph $G$ with $|C| = k$,
%with at least $N_{\ref{obs:robust_partition_colored_graphs}}(k, f, \gamma)$ vertices 
there exists an $(f, \gamma)$-robust equipartition $C' = C'_{\ref{obs:robust_partition_colored_graphs}}(C, f, 
\gamma)$ that refines $C$, where $|C'| \leq T$.
\end{observation}
\begin{proof}
Initially pick $C' = C$. Now, as long as $C'$ is not $(f, \gamma)$-robust, let $k'$ denote the number of parts of $C'$; we may replace $C'$ by a $G$-refinement $C''$ of it with at most $f(k')$ parts and
$\ind(C'') > \ind(C') + \gamma$.
This process stops after at most $1/\gamma$ iterations, by inequality \eqref{eq:index_refinement}.
\end{proof}

%%EF moved forward to here
The definition of robustness immediately implies the following.
\begin{observation}\label{obs:chainrobust}% Lemma 1 [Refinements and robustness]
Let $P = (V_1, \ldots, V_k)$ be an equipartition of a $\Sigma$-colored graph 
$G = (V, c)$, and suppose that $P$ is $(f \circ g, \gamma)$-robust for two functions $f,g:\mathbb{N} \to \mathbb{N}$ and $\gamma > 0$. Then any equitable refinement of $P$ with no
more than $g(k)$ parts is $(f, \gamma)$-robust.
\end{observation}

The notion of robustness is stronger than the more commonly used notion of regularity.  
For a $\Sigma$-colored graph $G = (V, c)$ and an equipartition $P = (V_1, \ldots, V_k)$ of $G$, 
a pair $(V_i, V_j)$ is \emph{$\epsilon$-regular} if $|d_\sigma(V_i, V_j) - d_\sigma(V'_i, V'_j)| \leq \epsilon$ for any $\sigma \in \Sigma$ and $V'_i \subseteq V_i, V'_j \subseteq V_j$ that satisfy $|V'_i| \geq \epsilon|V_i|, |V'_j| \geq \epsilon |V_j|$. $P$ is an \emph{$\epsilon$-regular partition} if all but at most $\epsilon \binom{k}{2}$
of the pairs $(V_i, V_j)$ are $\epsilon$-regular. 
The following lemma states that robust partitions are also regular; a lemma like it is implicit in the ideas of the original proof of \cite{Szemeredi1978}. The original was formulated only for non-colored graphs ($\Sigma=\{0,1\}$), but the extension to colored graphs is not hard (and was also done in prior work).
%%EF we forgot about $|\Sigma|$ dependencies everywhere - we need them!
\begin{lemma}[\cite{Szemeredi1978}, see also \cite{FischerNewman2005}]
\label{lem:robust_regular}
For any $\epsilon > 0$ there exist $f = f_{\ref{lem:robust_regular}}^{(\epsilon)}:\mathbb{N}\to \mathbb{N}$ and $\delta = \delta_{\ref{lem:robust_regular}}(|\Sigma|,\epsilon) > 0$ such that 
any $(f, \delta)$-robust equipartition $P$ of a $\Sigma$-colored graph $G$ is also $\epsilon$-regular.
\end{lemma}

The next lemma was first formulated (with different notation and without the extension to general $\Sigma$) in \cite{AlonFKS00}, but in a sense the basic idea was already used in implicitly proving Lemma \ref{lem:robust_regular} in \cite{Szemeredi1978}. It will be useful for us later.

\begin{lemma}[\cite{AlonFKS00}, see also \cite{FischerNewman2005}]
\label{lem:robust_to_densities}
For any $\epsilon > 0$ there exists $\delta = \delta_{\ref{lem:robust_to_densities}}(|\Sigma|,\epsilon) > 0$,
so that for every $f:\mathbb{N}\to \mathbb{N}$, any $(f, \delta)$-robust equipartition $P = (V_1, \ldots, V_k)$ of a $\Sigma$-colored graph $G$, and any equitable refinement $P' = (V_{11}, \ldots, V_{1b}, \ldots, V_{k1}, \ldots, V_{kb})$ of $P$ where $kb \leq f(k)$, choosing the indexes so that $V_i = \bigcup_{r=1}^{b} V_{ir}$ for any $i \in [k]$, it holds that
\[ \frac{1}{\binom{k}{2} b^2}
\sum_{\sigma \in \Sigma} \sum_{i, i' \in [k]} \sum_{j, j' \in [b]} |d_{\sigma}(V_{ij}, V_{i'j'}) - d_{\sigma}(V_i, V_{i'})| \leq \epsilon. 
\]
\end{lemma}

Another lemma that will be useful later is the following. This is Lemma 3.2 in $\cite{AlonFKS00}$.
\begin{lemma}[\cite{AlonFKS00}]
	\label{lem:regularity_to_many_copies}
	For any $\eta > 0$ there exists a function $\beta = \beta_{\ref{lem:regularity_to_many_copies}}^{\eta}:\mathbb{N} \to \mathbb{N}$,
	so that for any integer $l > 0$ there exists $\kappa = \kappa_{\ref{lem:regularity_to_many_copies}}(\eta, l)$ with the following property:
	If $H = ([l], c_H)$ and $V_1, \ldots, V_l$ are disjoint vertex sets of $G = (V,c)$, such that for any $i < j$,
	$(V_i, V_j)$ is $\beta(l)$-regular and $d_{c_H(ij)}(V_i, V_j) \geq \eta$, then the number of induced $H$-copies 
	in $G$ with a vertex from $V_i$ playing the role of vertex $i$ of $H$ is at least $\kappa \prod_{i=1}^{l} |V_i|$.
\end{lemma}

%%%%%%%%%%%%%%%%%%%%
\subsubsection*{Strings and intervals}
Consider an ordered set $V$ whose elements are $v_1 < \ldots < v_n$.
%Consider a \emph{string} $S:[n] \to \Gamma$. 
A \emph{string} $S:V \to \Sigma$ is a mapping from the ordered set $V$ to an alphabet $\Sigma$.
An \emph{interval partition} $I = (I_1, \ldots, I_k)$ of the string $S:V \to \Sigma$ is a partition $V = I_1 \ldots I_k$ into consecutive substrings of $S$: That is, there exist $0 = a_0 < \ldots < a_{k-1} < a_{k} = n$ such that $I_i = S(v_{a_{i-1}+1})  \ldots  S(v_{a_i})$ for any $1 \leq i \leq k$.
$I$ is \emph{equitable} (or an \emph{interval equipartition}) if $a_{i} - a_{i-1} \in \{ \lfloor n / k \rfloor, \lceil n / k \rceil \}$ for any $1 \leq i \leq k$. 
An \emph{interval refinement} $I'$ of $I$ is an interval partition of $S$ 
such that any part of $I'$ is fully contained in a part of $I$.
The \emph{size} $|I|$ of an interval partition $I$ is its number of parts.

Next we define notions of index and robustness that are suitable for strings and interval partitions.
Similar notions were established in \cite{AxenovichPersonPuzynina2013, FeigeKorenTennenholtz2017}.
For a string $S:V \to \Sigma$, the \emph{density} of $\sigma \in \Sigma$ in $S$ is $d_\sigma(S) = |S^{-1}(\sigma)|/|S|$
where $S^{-1}(\sigma) = \{v \in V: S(v) = \sigma\}$, and the squared density of $\sigma$ in $S$ is denoted by $d_\sigma^2(S)$.
The \emph{index} of $S$ is $\ind(S) = \sum_{\sigma \in \Sigma} d_\sigma^2(S)$.
Finally, the \emph{index} of an interval partition $I = (I_1, \ldots, I_k)$ of $S$ is
\[
\ind(I) = \sum_{i=1}^{k} \frac{|I_i|}{|V|}\ind(I_i).
\] 
As in the case of charts,
for an interval equipartition $I$ of a string $S$, 
we say that $I$ is $(f, \gamma)$-robust 
if any interval equipartition $I'$ of size at most $f(k)$ that refines $I$
satisfies $\ind(I') \leq \ind(I) + \gamma$ (in the other direction,
$\ind(I) \leq \ind(I')$ always holds). 
The following is an analogue of Observation \ref{obs:robust_partition_colored_graphs}, and its proof is essentially identical.
\begin{observation}[Robust partitioning of intervals]
\label{obs:robust_partition_intervals}
For any integer $k > 0$, function $f:\mathbb{N} \to \mathbb{N}$ and real $\gamma > 0$ there exists 
$T = T_{\ref{obs:robust_partition_intervals}}(k, f, \gamma)$ such that any interval equipartition $I$ of a string $S$ 
where $|I| = k$
has an $(f, \gamma)$-robust interval refinement $I' = I'_{\ref{obs:robust_partition_intervals}}(I, f, 
\gamma)$ consisting of at most $T$ intervals.
\end{observation}
The next lemma is an analogue of Lemma \ref{lem:robust_to_densities} for strings, and its proof is similar.
\begin{lemma}
\label{lem:string_robust_to_densities}
For any $\epsilon > 0$ there exists $\delta = \delta_{\ref{lem:string_robust_to_densities}}(|\Sigma|,\epsilon) > 0$,
so that for every $f:\mathbb{N}\to \mathbb{N}$, any $(f, \delta)$-robust interval equipartition $I = (I_1, \ldots, I_k)$ of a string $S:V \to \Sigma$, and any equitable interval refinement $I' = (I_{11}, \ldots, I_{1b}, \ldots, I_{k1}, \ldots, I_{kb})$ of $I$ where $kb \leq f(k)$, choosing the indexes so that $I_i = \bigcup_{r=1}^{b} I_{ir}$ for any $i \in [k]$, it holds that
\[ \frac{1}{kb}
\sum_{\sigma \in \Sigma} \sum_{i \in [k]} \sum_{j \in [b]} |d_{\sigma}(I_{ij}) - d_{\sigma}(I_i)| \leq \epsilon. 
\]
\end{lemma}

We finish by defining the string that a partition of an ordered set induces on that set. The strings that we will consider in this paper are of this type. 
\begin{definition}[String of a partition]
For a partition $P = (V_1, \ldots, V_k)$ of an ordered set $V$, the \emph{$P$-string} $S_P:V \to [k]$ maps any $v \in V$ to the element $i \in [k]$ such that $v \in V_i$. 
\end{definition}

With slight abuse of notation, we will also use the notion of an interval partition in the context of ordered graphs;
here each interval will simply be a set of consecutive vertices (with no accompanying function, in contrast to the case of strings).

\section{Technical aids}
We develop here two tools that we will use for our proofs. The first tool is a Ramsey-type theorem that we believe to be interesting in its own right. We will use it to find a ``uniform'' structure with a global view on our graph.

The second tool is a rounding lemma that allows us to evenly partition graphs also when the number of sets does not divide the number of vertices, without hand-waving away the divisibility issues (which might have been questionable in our context).

\subsection{A quantitative Ramsey-type theorem}\label{sub:ramsey}
The multicolored Ramsey number $\ram(s, k)$ is the minimum integer $n$ so that in any coloring of $K_n$ with $s$ colors
there is a monochromatic copy of $K_k$. It is well known that this number exists (i.e.\@ is finite) for any $s$ and $k$.
For our purposes, we will also need a different Ramsey-type result, that keeps track of ``undesirable'' edges, as described in the following subsection. 

Given a $k$-partite $\Sigma$-chart, we would like to pick a given number of vertices from each partition set, so that all edges between remaining vertices in each pair of sets are of the same color. 
However, in our situation we also have a ``quantitative'' requirement: A portion of the edges is marked as ``undesirable'', and we require that in the chart induced on the picked vertices the ratio of undesirable edges does not increase by much.

Formally, we prove the following, which we state as a theorem because we believe it may have uses beyond the use in this paper.

\begin{thm}\label{thm:qukramsey}
 There exists a function $R_{\ref{thm:qukramsey}}:\mathbb{N}\times\mathbb{N}\times\mathbb{N}\times (0,1]\to\mathbb{N}$, so that if $G=(V_1,\ldots,V_k,c)$ is a $k$-partite $\Sigma$-chart with $n\geq R_{\ref{thm:qukramsey}}(|\Sigma|,k,t,\alpha)$ vertices in each class, and $B\subseteq \bigcup_{i<j \in [k]}(V_i\times V_j)$ is a set of ``undesirable edges'' of size at most $\epsilon\binom{k}{2}n^2$, then $G$ contains an induced subchart $H_{\ref{thm:qukramsey}}(G,B,t,\alpha)=(W_1,\ldots,W_k,c\restrict{\bigcup_{1\leq i<j\leq k}(W_i\times W_j)})$ with the following properties.
 \begin{itemize}
  \item $|W_i|=t$ for every $1\leq i\leq k$.
  \item $c\restrict{W_i\times W_j}$ is a constant function for every $1\leq i<j\leq k$.
  \item The size of $B\cap (\bigcup_{1\leq i<j\leq k}(W_i\times W_j))$ is at most $(1+\alpha)\epsilon\binom{k}{2}t^2$.
 \end{itemize}
\end{thm}

In our use, these ``vertices'' would actually be themselves sets of a robust partition of the original graph, and ``colors'' will encode densities; an undesirable pair would have the ``wrong'' densities. Also, in our use case the undesirability of an edge will be determined solely by its color and the $W_i$ that its end vertices belong to, which means that for each $1\leq i<i'\leq k$ the edge set $W_i\times W_{i'}$ consist of either only desirable edges or only undesirable edges. When this happens, a later pick of smaller sets $W'_i\subset W_i$ will still preserve the ratio of undesirable edges (we will in fact perform such a pick using the original Ramsey's theorem inside each $W_i$).
The following corollary summarizes our use of the theorem.
\begin{definition}
Given a $k$-partite $\Sigma$-chart $G=(V_1,\ldots,V_k,c)$ and a set $B\subseteq \bigcup_{i<j \in [k]}(V_i\times V_j)$, we say that $B$ is {\em orderly} if for every $1\leq i<j\leq k$ there are no $e\in(V_i\times V_j)\cap B$ and $e'\in(V_i\times V_j)\setminus B$ for which $c(e)=c(e')$. In other words, the ``position'' and color of an edge fully determines whether it is in $B$.
\end{definition}

\begin{corollary}
\label{coro:ramsey_app}
There exists a function $R_{\ref{coro:ramsey_app}}:\mathbb{N} \times \mathbb{N} \times \mathbb{N} \to \mathbb{N}$, so that if
$G = (\bigcup_{i=1}^{k} V_i, c)$ is a $\Sigma$-colored graph with $|V_i| = n \geq R_{\ref{coro:ramsey_app}}(|\Sigma|, k, t)$ for any $i \in [k]$ and $V_i \cap V_j = \emptyset$ for any $i \neq j \in [k]$,
and $B \subseteq \bigcup_{i < j \in [k]} (V_i \times V_j)$ is an orderly set of ``undesirable edges'' of size at most $\epsilon \binom{k}{2} n^2$,
then $G$ contains an induced subgraph $D$ satisfying the following.
\begin{itemize}
\item The vertex set of $D$ is $\bigcup_{i=1}^{l} U_i$ where $U_i \subseteq V_i$ and $|U_i| = t$ for any $i \in [k]$. 
\item For any $i \in [k]$, all edges inside $U_i$ have the same color.
\item For any $i < j \in [k]$, all edges in $U_i \times U_j$ have the same color.
\item $\sum_{i < j \in [k]} |B \cap (U_i \times U_j)|\leq 2  \epsilon \binom{k}{2} t^2$.
\end{itemize}
\end{corollary}
\begin{proof}
Take $R_{\ref{coro:ramsey_app}}(s, k, t) = R_{\ref{thm:qukramsey}}(s, k, \ram(s, t), 1)$ (recall that $\ram(s,t)$ denotes the ``traditional'' $s$-colored Ramsey function).
By Theorem \ref{thm:qukramsey}, there exists a chart $H = (W_1, \ldots, W_k)$ with the following properties.
\begin{itemize}
\item $W_i \subseteq V_i$ and $|W_i| = \ram(t, |\Sigma|)$ for every $i \in [k]$.
\item For any pair $i < j \in [k]$, all edges in $W_i \times W_j$ have the same color.
\item $\sum_{i < j \in [k]} |B \cap (W_i \times W_j)|\leq 2  \epsilon \binom{k}{2} \left(\ram(t, |\Sigma|)\right)^2$.
\end{itemize}
Observe that for any pair $i < j \in [k]$, either $W_i \times W_j \subseteq B$ or $(W_i \times W_j) \cap B = \emptyset$, since all edges in $W_i \times W_j$ have the same color and $B$ is orderly. Therefore, the number of pairs $i < j$ for which $(W_i \times W_j) \cap B \neq \emptyset$ is at most $2 \epsilon \binom{k}{2}$.
Now we apply the traditional Ramsey's theorem inside each $W_j$ to obtain a set $U_j \subseteq W_j$ of size $t$ such that all edges inside $W_j$ have the same color.
Since $\sum_{i < j} |B \cap (U_i \times U_j)| \leq \sum_{i<j:(W_i \times W_j) \cap B \neq \emptyset} |B \cap (U_i \times U_j)| \leq 2 \epsilon \binom{k}{2} t^2$, the proof follows.
\end{proof}

Before moving to the proof of Theorem \ref{thm:qukramsey} itself, let us quickly note that a quantitative counterpart for the traditional (not $k$-partite) graph case does not exist (indeed, Corollary \ref{coro:ramsey_app} is a way for us to circumvent such issues).

\begin{proposition}\label{prop:noqutramsey}
 For any $\alpha>0$, $m$, $k$, and large enough $l$, for infinitely many $n$ there is a graph $G$ and a set of undesirable pairs $B$, so that $G$ has $n$ vertices, $B$ consists of at most $\frac{1}{mk}\binom{n}{2}$ pairs, $G$ has no independent set of size $l$, and every clique of $l$ vertices in $G$ holds at least $(\frac{1}{m}-\alpha)\binom{l}{2}$ members of $B$.
\end{proposition}

\begin{proof}
 We construct $G$ for any $n$ that is a multiple of $mk$ larger than $lk$. The graph $G$ will be the union of $k$ vertex-disjoint cliques, each with $n/k$ vertices. In particular, $G$ contains no independent set with $l$ vertices, and any clique with $l$ vertices must be fully contained in one of the cliques of $G$.
 
 Now $B$ will be fully contained in the edge-set of $G$, and will consist of the edge-set of $mk$ vertex-disjoint cliques with $n/mk$ vertices each, so that each of the cliques of $G$ contains $m$ of them. It is now not hard to see that any clique with $l$ vertices in $G$ will contain at least  $(\frac{1}{m}-\alpha_l)\binom{l}{2})$ members of $B$, where $\lim_{l\to\infty}\alpha_l=0$.
\end{proof}

Moving to the proof, the following is our main lemma. It essentially says that we can have a probability distribution over ``Ramsey-configurations'' in our chart that has some approximate uniformity properties.

\begin{lemma}\label{lem:probramsey}
 There exists a function $R_{\ref{lem:probramsey}}:\mathbb{N}\times\mathbb{N}\times\mathbb{N}\times (0,1]\to\mathbb{N}$, so that if $G=(V_1,\ldots,V_k,c)$ is a $k$-partite $\Sigma$-chart with $n\geq R_{\ref{lem:probramsey}}(|\Sigma|,k,t,\delta)$ vertices in each class, then $G$ contains a randomized induced subchart $H_{\ref{lem:probramsey}}(G,t,\delta)=(W_1,\ldots,W_k,c\restrict{\bigcup_{1\leq i<j\leq k}(W_i\times W_j)})$ satisfying the following properties.
 \begin{itemize}
  \item Either $|W_i|=t$ for every $1\leq i\leq k$, or the chart is empty ($W_i=\emptyset$ for every $i$).
  \item $c\restrict{W_i\times W_j}$ is a constant function for every $1\leq i<j\leq k$ (with probability $1$).
  \item For every $1\leq i\leq k$, every $v\in V_i$ will be in $W_i$ with probability at most $t/n$.
  \item For every $1\leq i<j\leq k$, every $v\in V_i$ and every $w\in V_j$, the probability for both $v\in W_i$ and $w\in W_j$ to hold is bounded by $(t/n)^2$.
  \item The probability that the chart is empty is at most $\delta$.
 \end{itemize}
\end{lemma}

Before we prove this lemma, we show how it implies Theorem \ref{thm:qukramsey}.

\begin{proof}[Proof of Theorem \ref{thm:qukramsey}]
 We set $R_{\ref{thm:qukramsey}}(a,k,t,\alpha)=R_{\ref{lem:probramsey}}(a,k,t,\alpha/3)$. Given the $k$-partite $\Sigma$-chart $G$, we take the randomized subchart $H=H_{\ref{lem:probramsey}}(G,t,\alpha/3)=(W_1,\ldots,W_k,c\restrict{\bigcup_{1\leq i<j\leq k}(W_i\times W_j)})$, and prove that with positive probability it is the required subchart.
 
 Let $B'=B\cap(\bigcup_{1\leq i<j\leq k}W_i\times W_j)$ denote the set of undesirable pairs that are contained in $H$. By the probability bound on pair containment and by the linearity of expectation, $\expe [|B'|]\leq (t/n)^2|B|\leq \epsilon\binom{k}{2}t^2$. Therefore, the probability for $|B'|$ to be larger than $(1+\alpha)\epsilon\binom{k}{2}t^2$ is bounded by $\frac{1}{1+\alpha}\leq 1-\alpha/2$. Therefore, with positive probability, both $|B'|$ is not too large and $H$ is not the empty chart. Such an $H$ is the desired subchart.
\end{proof}

To prove Lemma \ref{lem:probramsey} we shall make good use of the following near-trivial observation.

\begin{observation}\label{obs:pick}
 There exists a function $m_{\ref{obs:pick}}:\mathbb{N}\times\mathbb{N}\times (0,1]\to\mathbb{N}$, such that if $A$ is a set of size at least $m_{\ref{obs:pick}}(k,t,\delta)$ and $\mathcal{A}=(A_1,\ldots,A_k)$ is a partition of $A$ to $k$ sets, then there exists a randomized subset $B=B_{\ref{obs:pick}}(\mathcal{A},t,\delta)$ satisfying the following properties.
 \begin{itemize}
  \item Either $|B|=t$ or $B=\emptyset$.
  \item $B$ is fully contained in a single $A_i$.
  \item For every $a\in A$, the probability for $a\in B$ is at most $t/|A|$.
  \item The probability for $B=\emptyset$ is at most $\delta$.
 \end{itemize}
\end{observation}

\begin{proof}
 To choose the randomized subset $B$, first choose a random index $I$ where $\pr [I=i]=|A_i|/|A|$ for all $1\leq i\leq k$. Next, if $|A_I|<t$ then set $B=\emptyset$, and otherwise set $B$ to be a subset of size exactly $t$ of $A_I$, chosen uniformly at random from all $\binom{|A_I|}{t}$ possibilities. Setting $m_{\ref{obs:pick}}(k,t,\delta)=tk/\delta$, it is not hard to see that all properties for the random set $B$ indeed hold.
\end{proof}

\begin{proof}[Proof of Lemma \ref{lem:probramsey}]
 The proof is done by induction over $k$. The base case $k=1$ is easy -- set $R_{\ref{lem:probramsey}}(|\Sigma|,1,t,\delta)=t$, and let $W_1$ be a uniformly random subset of size $t$ of $V_1$.
 
 For the induction step from $k-1$ to $k$, we set $R_{\ref{lem:probramsey}}(|\Sigma|,k,t,\delta)=m_{\ref{obs:pick}}(|\Sigma|^s,r,\frac{1}{k+1}\delta)$, where $s=m_{\ref{obs:pick}}(|\Sigma|^{k-1},t,\frac{1}{k+1}\delta)$ and $r=R_{\ref{lem:probramsey}}(|\Sigma|,k-1,t,\frac{1}{k+1}\delta)$. We set $W_1,\ldots,W_k$ to be the result of the following random process.
 
 First, we set $V'_1\subseteq V_1$ to be a uniformly random subset of size exactly $s$. Then, for every $2\leq i\leq k$, we set $V'_i\subseteq V_i$ to be the random set $B_{\ref{obs:pick}}(\mathcal{V}_i,r,\frac{1}{k+1}\delta)$, where $\mathcal{V}_i$ is the partition of $V_i$ obtained by classifying every $v\in V_i$ according to the colors $\langle c(w,v)\rangle_{w\in V'_1}$, i.e., two vertices in $V_i$ are in the same partition set if all their pairs with vertices from $V'_1$ have the same colors. 
 
 If any of the $V'_i$ came out empty, we set all $W_i$ to $\emptyset$ and terminate the algorithm (this occurs with probability at most $\frac{k-1}{k+1}\delta$), and otherwise we continue. Note now, in particular, that for every $w\in V_1$ and $v\in V_i$ the probability for both $w\in V'_1$ and $v\in V'_i$ to hold is bounded by $(s/n)(r/n)$. This is since the probability guarantees of Observation \ref{obs:pick} hold for any possible value of $V'_1$. Also, since each $V'_i$ was independently drawn, for $v\in V_i$ and $w\in V_j$ (where $1<i<j\leq k$) the probability for both $v\in V'_i$ and $w\in V'_j$ to hold is bounded by $(r/n)^2$.
 
 We now let $H'$ denote the $(k-1)$-partite $\Sigma$-chart induced by $V'_2,\ldots,V'_k$, and use the induction hypothesis to (randomly) set $W_2,\ldots,W_k$ as the corresponding vertex sets of $H_{\ref{lem:probramsey}}(H',t,\frac{1}{k+1}\delta)$. As before, if we receive empty sets then we also set $W_1=\emptyset$ and terminate. Note that for $1<i\leq k$ and $v\in V'_i$, the probability for $v$ to be in $W_i$ is bounded by $t/r$. Hence, for $v\in V_i$, the probability for $v\in W_i$ to hold is bounded by $(r/n)(t/r)=t/n$. Similarly, for $1<i<j\leq k$, for every $v\in V_i$ and $w\in V_j$ the probability for both $v\in W_i$ and $w\in W_j$ to hold is bounded by $(t/n)^2$. Also by similar considerations, for $v\in V_1$ and $w\in V_i$, the probability for both $v\in V'_1$ and $w\in W_i$ to hold is bounded by $(s/n)(t/n)$.
 
 Finally, we set $W_1$ to be the random set $B_{\ref{obs:pick}}(\mathcal{V}',t,\frac{1}{k+1}\delta)$, where $\mathcal{V}'$ is the partition of $V'_1$ obtained by classifying each $v\in V'_1$ by the colors $\langle c(v,w)\rangle_{w\in W_i}$. Note that $c(v,w)$ in that expression depends only on $v$ and the index $i$ for which $w\in W_i$, because of how we chose each $V'_i$ above. In particular, after the choice of $W_1$, the function $c\restrict{W_1\times W_i}$ is constant for each $1<i\leq k$. Again, if we got an empty set for $W_1$, we set all $W_2,\ldots,W_k$ to be empty as well. By similar considerations as in the preceding steps, also here, for any $v\in V_1$ the probability of $v\in W_1$ is bounded by $t/n$, and for $w\in V_i$ where $1<i\leq k$, the probability of both $v\in W_1$ and $w\in W_i$ is bounded by $(t/n)^2$.
 
 The probability of obtaining empty sets in any of the steps is bounded by $\delta$ by a union bound, and all other properties of the random sets $W_1,\ldots,W_k$ have already been proven above.
\end{proof}

\subsection{Multipartitions and rounding}\label{sub:rounding}

The following is a mechanism to handle ``with one stroke'' rounding issues throughout the paper.

\begin{definition}[Multipartitions]
A {\em multipartition} of a set $L$ is a family $M$ of subsets of $L$, that in particular includes $L$ and all the singletons $\{i\}$ for $i\in L$, and furthermore every two sets $A,B\in M$ are either disjoint or one is contained in the other.
\end{definition}

To get an idea, an object that can be modeled as a multipartition is a partition of $L$ (the multipartition would contain the partition sets, along with $L$ itself and all singleton sets $\{i\}$), but also other objects, such as a partition and its refinement together, can be modeled as multipartitions. Here is the main lemma.

\begin{lemma}[rounding feasibility]\label{lem:round}
If $M$ and $N$ are two multipartitions of the same set $L$, and $\lambda_i\in\mathbb{R}$ is a real value attached to every $i\in L$, then there exist integer values $\ell_i\in\mathbb{Z}$ attached to $i\in L$, satisfying the following.
\begin{itemize}
 \item $\ell_i\in\{\lfloor\lambda_i\rfloor,\lceil\lambda_i\rceil\}$ for every $i\in L$.
 \item $\sum_{i\in A}\ell_i\in\{\lfloor\sum_{i\in A}\lambda_i\rfloor,\lceil\sum_{i\in A}\lambda_i\rceil\}$ for every $A\in M$ and for every $A\in N$.
 \item $\sum_{i\in L}\ell_i\in\{\lfloor\sum_{i\in L}\lambda_i\rfloor,\lceil\sum_{i\in L}\lambda_i\rceil\}$.
\end{itemize}
\end{lemma}

\begin{proof}
Note that the middle item implies the other two (since $M$ and $N$ in particular include $L$ and all singleton sets $\{i\}$ for $i\in L$).
We define the following problem of solving a flow network with minimal and maximal constraints
(for an exposition of flow networks see \cite{VanLeeuwen1990}).
\begin{itemize}
 \item The node set of the flow network is $\{u_A:A\in M\}\cup\{w_A:A\in N\}$.
 \item The start node is $u_L$ and the target node is $w_L$.
 \item For every $A,B\in M$, so that $A\subsetneq B$ and there is no $C\in M$ for which $A\subsetneq C\subsetneq B$, we put an edge from $u_B$ to $u_A$ with minimum flow $\lfloor\sum_{i\in A}\lambda_i\rfloor$ and maximum flow $\lceil\sum_{i\in A}\lambda_i\rceil$.
 \item For every $A,B\in N$, so that $A\subsetneq B$ and there is no $C\in N$ for which $A\subsetneq C\subsetneq B$, we put an edge from $w_A$ to $w_B$ with minimum flow $\lfloor\sum_{i\in A}\lambda_i\rfloor$ and maximum flow $\lceil\sum_{i\in A}\lambda_i\rceil$.
 \item For every $i\in L$ we put an edge from $u_i$ to $w_i$ with minimum flow $\lfloor\lambda_i\rfloor$ and maximum flow $\lceil\lambda_i\rceil$.
 \item We require the total flow of the network to be between $\lfloor\sum_{i\in L}\lambda_i\rfloor$ and $\lceil\sum_{i\in L}\lambda_i\rceil$.
\end{itemize}
This flow network has a real-valued solution by assigning $\lambda_i$ flow to each edge of the type $u_i,w_i$, and then assigning the corresponding sums to all other network edges. Hence (since all constraints are integer-valued), the flow network has an integer-valued solution as well (see, e.g., the analysis of Lawler's algorithm in \cite{VanLeeuwen1990}, page 602). Fixing such a solution, and setting $\ell_i$ to be the flow in the edge $u_i,w_i$ for every $i\in L$, completes the proof.
\end{proof}

An example of using the lemma is when we want to round the values in a $2$-dimensional matrix so that the row sums and column sums are also rounded versions of the original sums (and in particular equal to the original sums if they happen to be integers). In our use the resulting integer values would be set sizes for an equipartition, that in turn refines other partitions with set size requirements.

We also note here that the statement of this lemma is false when we are presented with three multipartitions. Take for example the $3$-dimensional matrix of size $2\times 2\times 2$, where $\lambda_{111}=\lambda_{100}=\lambda_{010}=\lambda_{001}=\frac12$, with all other $\lambda$ values being zero. Also for each of the three dimensions take the partition into two axes-parallel planes. The values on every set of every partition sum up to $1$, and yet there are no corresponding $\ell_{ijk}\in\{\lfloor\lambda_{ijk}\rfloor,\lceil\lambda_{ijk}\rceil\}$ satisfying these constraints.

%%%%%%%%%%%%%%%%%%%%%%%%%%%%%%%%%%%%%%%%%%%%%%%%%%%%%%%%%%%%%%%%%%%%%%%%%%%%%%%%
\section{A regularity scheme for ordered graphs}\label{sec:regscheme}

\subsection{The approximating partition framework}\label{sub:approx}

\begin{definition}[$\delta$-approximating partitions]
\label{def:delta_approx}
Let $P = (V_1, \ldots, V_k)$ and $Q = (U_1, \ldots, U_l)$ be partitions of a set $V$ of size $n$.
We say that $Q$ is a \emph{$\delta$-approximation} of $P$, or equivalently, that $P$ and $Q$ are \emph{$\delta$-close},
if there exists $T \subseteq V$ with $|T| \leq \delta n$ such that $V_i \setminus T = U_i \setminus T$ for any $1 \leq i \leq \max\{k,l\}$,
where for $i>k$ we define $V_i = \phi$, and similarly $U_i = \phi$ for $i > l$.
\end{definition}
%Assuming w.l.o.g.\@ that $k < l$ in the definition, a simple counting argument implies that $U_i \subseteq T$ for any $i > k$.

\begin{lemma}% Lemma 2 [Approximations and index]
\label{lem:approx_index_1}
For any $\epsilon > 0$ there exists $\delta = \delta_{\ref{lem:approx_index_1}}(\epsilon) > 0$
such that
any two $\delta$-close partitions $P$ and $Q$ of (the vertex set of) a colored graph $G=(V,c)$ satisfy $|\ind(P) - \ind(Q)| \leq \epsilon$.
\end{lemma}

\begin{proof}
Let $P = (V_1, \ldots, V_k)$ and $Q = (U_1, \ldots, U_l)$ be $\delta$-close partitions of $G$, where we assume w.l.o.g. that $k < l$.
For any $1 \leq i \leq k$  let $W_i = V_i \cap U_i$, and observe that $\sum_{i = 1}^{k} |W_i| \geq (1-\delta) n$.
we say that $i$ is \emph{bad} if $|W_i| \leq (1 - \sqrt{\delta}) \min\{|V_i|, |U_i|\}$ and \emph{good} otherwise.
Then $\sum_{i \text{ bad}} |V_i| \leq \sqrt{\delta} n$ and %, as $\sum_{i = 1}^{k} |V_i| - |W_i| \leq \delta n$. Similarly, 
$\sum_{i \text{ bad}} |U_i| \leq \sqrt{\delta} n$. 
%
%Take 
%\[S = \sum_{\overset{1 \leq i < j \leq k}{i,j \text{ good}}} \frac{|W_i| |W_j|}{\binom{n}{2}} d_\sigma^2 (W_i, W_j).\] 
%We will show that $|\ind(P) - S| = O(\sqrt{\delta})$ and
%get, by symmetry, that $|\ind(Q) - S| = O(\sqrt{\delta})$, which will finish the proof by taking $\delta = \Theta(\epsilon^2)$.
When $i$ and $j$ are both good, we have
\begin{align*}
&\big|\ind(V_i, V_j) - \ind(W_i, W_j)\big| \leq \sum_{\sigma \in \Sigma} \big|(d_\sigma^2(V_i, V_j) - d_\sigma^2(W_i, W_j)\big| \leq
2 \sum_{\sigma \in \Sigma} \big|d_\sigma(V_i, V_j) - d_\sigma(W_i, W_j)\big|  \\
&\leq 2\sum_{\sigma \in \Sigma} \left(  d_\sigma(V_i, V_j) \left(\frac{|V_i| |V_j|}{|W_i||W_j|} -1\right) 
+ \frac{|c^{-1}(\sigma) \cap ((V_i \times V_j) \setminus (W_i \times W_j))|}{|W_i| |W_j|}  \right) 
 = O(\sqrt{\delta})
\end{align*}
where the second inequality holds since $d_\sigma(V_i, V_j) + d_\sigma(U_i, U_j) \leq 2$, the third inequality follows from
the fact that $|x-y| \leq z-x + z-y$ for $z \geq \max\{x,y\}$ and the last inequality follows from the observation that $|V_i| |V_j| = (1+O(\sqrt{\delta})) |W_i| |W_j|$. 
Similarly, it holds that $|\ind(U_i, U_j) - \ind(W_i, W_j)| = O(\sqrt{\delta})$, so 
$|\ind(V_i, V_j) - \ind(U_i, U_j)| = O(\sqrt{\delta})$ when $i,j$ are good.
We finish by observing that 
\begin{align*}
\ind(P) - \ind(Q)
&\leq \sum_{i<j \text{ good}} \left( \frac{|V_i||V_j|}{\binom{n}{2}} \ind(V_i, V_j) - \frac{|U_i||U_j|}{\binom{n}{2}} \ind(U_i, U_j)\right) + 2\sqrt{\delta} = O(\sqrt{\delta})
\end{align*}
where the first inequality holds since $\sum_{i \text{ bad}} \sum_{j \neq i}|V_i||V_j| ind(V_i, V_j) / \binom{n}{2} \leq 2\sqrt{\delta}$ and
the second inequality is true since $|V_i||V_j| = \left(1 + O(\sqrt{\delta})\right)|U_i||U_j|$ and $\ind(V_i, V_j) = \left(1+O(\sqrt{\delta})\right) \ind(U_i, U_j)$ when $i$ and $j$ are good, and since $\ind(Q) \leq 1$. Therefore, taking a suitable $\delta = \Theta(\epsilon^2)$ in the statement of the lemma suffices.
\end{proof}

\begin{lemma}
\label{lem:appx_ref1}
% Lemma 3 [Approximations and refinements]
%Let $P,Q$ be $\delta$-close (equitable) partitions of a colored graph $G$.
%Then any (equitable) refinement $Q'$ of $Q$ is a $\delta$-approximation of some (equitable) refinement $P'$ of $P$.
Let $P,Q$ be $\delta$-close equipartitions of a colored graph $G$, where $|P| = |Q|$.
Then any equitable refinement $Q'$ of $Q$ is $\delta$-close to an equitable refinement $P'$ of $P$, with $|P'| = |Q'|$. 
\end{lemma}
\begin{proof}
Write $P = (V_1, \ldots, V_k), Q = (U_1, \ldots, U_k), Q' = (U_{11}, \ldots, U_{1 r}, \ldots, U_{k 1}, \ldots, U_{k r})$ where $U_i = \bigcup_{j=1}^r U_{ij}$. Also, for any $i,j$ let $W_i = V_i \cap U_i$ and $W_{ij} = V_i \cap U_{ij}$. Then $\sum_{i=1}^{k} \sum_{r=1}^r |W_{ij}| = \sum_{i=1}^{k} |W_i| \geq (1 - \delta)n$, so we may take a refinement $P' = (V_{ij})_{1 \leq i \leq k, 1 \leq j \leq r}$ of $P$ as follows: For any $i,j$ we take $V_{ij}$ that contains 
$W_{ij}$ and $\ell_{ij}$ arbitrary additional elements from $V_i \setminus W_i$, where $\ell_{ij}$ is chosen by using Lemma \ref{lem:round} in the following manner.

For $1\leq i\leq k$ and $1\leq j\leq r$ we set $\lambda_{ij}=|V_i|/r-|W_{ij}|$. We set the multipartition $M$ to consist of all singleton sets $\{ij\}$, the set $[k]\times [r]$, and the sets $\{i\}\times [r]$ for $1\leq i\leq k$. We set the multipartition $N$ to be the trivial one, just the singleton sets and  $[k]\times [r]$. Invoking Lemma \ref{lem:round}, we claim the following about the resulting $\ell_{ij}$: Since $\sum_{j=1}^r\lambda_{ij}=|V_i|-|W_i|$, which is an integer, this will also equal the corresponding sum $\sum_{j=1}^r\ell_{ij}=|V_i|-|W_i|$, so we can get this way a refinement of $P$. Also, for any integer $m$ (in our case $|V_i|$) it holds that $\lfloor\frac{m}{r}\rfloor=\lceil\frac{m+1}{r}\rceil-1$, so the resulting $V_{ij}$ would form an equipartition. The last issue that we need to deal with is when we have $\ell_{ij}=-1$ for some $i$ and $j$, which could in fact happen. We claim however that in such a case we can move to another solution for which $\ell_{ij}=0$. To see this, we note that $\ell_{ij}=-1$ only if $|V_i|/r$ is not an integer, $|W_{ij}|=|U_{ij}|=\lceil|V_i|/r\rceil$, and $\ell_{ij}=\lfloor\lambda_{ij}\rfloor$. But in this case one can see that there exists $j'\neq j$ so that $\ell_{ij'}=\lceil\lambda_{ij'}\rceil>0$, and so we can increase $\ell_{ij}$ by $1$ at the expense of $\ell_{ij'}$.
%If $|U_i| = |V_i|$ then take $v_{ij} = |U_{ij}| - |W_{ij}|$.
%Otherwise, suppose that $|U_i| = |V_i|-1$; the case $|U_i| = |V_i| + 1$ is symmetric.
%Pick some $1 \leq m_i \leq r$ for which $|U_{im_i}|$ is minimal.  Take $v_{im_i} = |U_{im_i}| - |W_{im_i}| + 1$ and $v_{ij} = |U_{ij}| - |W_{ij}|$ for any $j \neq m_i$. In any case, it is clear that $P'$ and $Q'$ are $\delta$-close and that $|P'| = |Q'|$.
%Moreover, $P'$ is an equipartition, as can be verified by the reader. 
\end{proof}

\begin{lemma}% Lemma 4 [Approximations and robustness]
\label{lem:approx_robust1}
For any $\epsilon > 0$ there exists $\delta = \delta_{\ref{lem:approx_robust1}}(\epsilon) > 0$ such that the following holds: 
If $P$ and $Q$ are $\delta$-close equipartitions of a colored graph $G$, $P$ is $(f, \delta)$-robust and $|P| = |Q|$,
then $Q$ is $(f, \epsilon)$-robust.
\end{lemma}
\begin{proof}
Let $P, Q$ be equipartitions as in the statement and pick $\delta = \delta_{\ref{lem:approx_index_1}}(\epsilon/3)$.
%\min\{\epsilon/3, \delta_{\ref{lem:approx_index_1}}(\epsilon/3)\}$.
Consider an equitable refinement $Q'$ of $Q$ of size at most $f(|Q|) = f(|P|)$. By Lemma 
\ref{lem:appx_ref1} there exists some equitable refinement $P'$ of $P$
which $\delta$-approximates $Q'$ where $|P'| = |Q'| \leq f(|P|)$. 
The robustness of $P$ implies that $\ind(P') - \ind(P) \leq \delta \leq \epsilon / 3$.
By Lemma \ref{lem:approx_index_1}, $|\ind(P) - \ind(Q)| \leq \epsilon / 3$ and $|\ind(P') - \ind(Q')| \leq \epsilon / 3$. We conclude that 
$\ind(Q') - \ind(Q) \leq \epsilon$. Thus, $Q$ is $(f, \epsilon)$-robust.
\end{proof}

The definition of $\delta$-close partitions works exactly the same for interval partitions. We observe that interval equipartitions and their densities are mostly determined by the number of parts.
\begin{observation}\label{obs:allintapprox}
Any two interval equipartitions $I$ and $J$ of $[n]$ into $m$ parts are $\frac{m^2}{n}$-close to each other. In particular, for any $f$, $m$ and $\epsilon$, if $n$ is large enough as a function of $m$ and $\epsilon$, then for such $I$ and $J$ the densities satisfy $\frac1m\sum_{i=1}^m\sum_{\sigma\in\Sigma}|d_{\sigma}(I_i)-d_{\sigma}(J_i)|\leq\epsilon$.
\end{observation}
\begin{proof}
If $I$ and $J$ are two interval equipartitions of $[n]$ with $|I|=|J|=m$, then we can set $T=\bigcup_{i=1}^{m-1}[i\lfloor\frac{n}{m}\rfloor+1,i\lceil\frac{m}{m}\rceil]$. Clearly $|T|<m^2$ and $I_i\setminus T=J_i\setminus T$ for every $i\in [m]$. The second part of the observation then follows easily from the first part for $n$ large enough.
\end{proof}

\subsection{The core lemmas}\label{sub:lcrscheme}

\begin{definition}[Least Common Refinement]
For two partitions $P = (V_1, \ldots, V_k)$ and $Q = (U_1, \ldots, U_l)$ of a colored graph $G$,
the \emph{least common refinement (LCR)} $P \sqcap Q$ of $P$ and $Q$ is the partition 
$(V_1 \cap U_1, \ldots, V_1 \cap U_l, \ldots, V_k \cap U_1, \ldots, V_k \cap U_l)$ (after removing empty sets from the list).
\end{definition}

Note that even if $P$ and $Q$ are equitable, $P \sqcap Q$ is not necessarily equitable.

The following lemma allows us to combine an ``order-respecting'' interval partition and a robust graph partition. The last statement in the formulation (about even $n$ and $m$) is not needed for the rest of the proofs concerning ordered graphs, but we will refer to it when we discuss ordered matrices.

\begin{lemma}%[Refinements whose intersection is an equipartition]% Lemma 6
\label{lem:main_equiparts}
For any $\delta > 0$ and positive integers $k$, $m$ and $b$, there exists $\gamma = \gamma_{\ref{lem:main_equiparts}}(\delta,k) > 0$ such that the following holds: 
If $P$ is an equipartition of an $n$-vertex colored graph $G$ (for $n\geq N_{\ref{lem:main_equiparts}}(\delta,k,m,b)$) with $|P| = k$,
and $J$ is an interval equipartition of $S_P$ of size $m$ which is $(f, \gamma)$-robust, where $f(m)\geq mb$,
% and $n \geq N_{\ref{lem:main_equiparts}}(\delta,k,m,b)$
then there exist an interval equipartition $I=I_{\ref{lem:main_equiparts}}(\delta,P,m,b)$ of size $m$,
an interval equipartition $I'=I'_{\ref{lem:main_equiparts}}(\delta,P,m,b)$ of size $mb$ which refines $I$,
an equipartition $Q = Q_{\ref{lem:main_equiparts}}(\delta,P,m,b)$ of size $k$ which $\delta$-approximates $P$, and 
an equipartition $Q' = Q'_{\ref{lem:main_equiparts}}(\delta,P,m,b)$ of size at most $T_{\ref{lem:main_equiparts}}(\delta,k,m)$ which is a refinement of both
$I$ and $Q$, all satisfying that the LCR $Q'' =  I' \sqcap Q'$ is an equipartition of size $|Q''|=b|Q'|=|Q'||I'|/|I|$ (i.e., each set of $Q'$ intersects ``nicely'' all subintervals of the interval of $I$ that contains it).

Furthermore, if $m$ and $n$ are even, then $I$ ``respects the middle'', that is $\sum_{i=1}^{m/2}|I_i|=\frac{n}{2}$.
\end{lemma}

\begin{proof}
We denote $P=(V_1,\ldots,V_k)$, and set $\gamma_{\ref{lem:main_equiparts}}(\delta,k)=\delta_{\ref{lem:string_robust_to_densities}}(k,\delta/20)$. The sets of the original interval equipartition $J$ will be denoted by $J_1,\ldots,J_m$.

We denote the eventual intervals of $I$ by $(I_1,\ldots,I_m)$, denote the eventual intervals of $I'$ by $(I_{11}, \ldots, I_{1b}, \ldots, I_{m1}, \ldots, I_{mb})$  where $I_i = \bigcup_{j=1}^{b} I_{ij}$ for any $i \in [m]$. The eventual sets of $Q$ will be denoted by $(U_1,\ldots,U_k)$, the sets of $Q'$ by $(W_{11},\ldots,W_{1t},\ldots,W_{m1},\ldots,W_{mt})$ where $I_i=\bigcup_{s=1}^tW_{is}$, and the eventual sets of $Q''$ will be denoted as $W_{ijs}=W_{is}\cap I_{ij}$. We pick $t=k\lceil 20/\delta\rceil$, and correspondingly $T_{\ref{lem:main_equiparts}}(\delta,k,m)=mt$.

Before choosing the partition intervals and sets themselves, we will choose sizes for the sets, and also choose sets of indexes $K_1,\ldots,K_k$ describing the connection of $Q$ to its refinement $Q'$. That is, eventually we will have $U_a=\bigcup_{(is)\in K_a}W_{is}$ for every $a\in [k]$. Finally defining the eventual $K_{ia}=\{s:(is)\in K_a\}$, $U_{ia}=I_i\cap U_a$ and $U_{ija}=I_{ij}\cap U_a$, we will also have $U_{ia}=\bigcup_{s\in K_{ia}}W_{is}$ and $U_{ija}=\bigcup_{s\in K_{ia}}W_{ijs}$.

We next determine the sizes $|K_{ia}|$, which will be found through our first use of Lemma \ref{lem:round} (and some further processing). We set the following parameters and multipartitions.
\begin{itemize}
 \item $\lambda_{ia}=t|J_i\cap V_a|/|J_i|$.
 \item $N$ contains the singleton sets $\{(ia)\}$, the set $[m]\times [k]$, and the set $\{i\}\times [k]$ for every $i\in [m]$. Note that in particular $\sum_{a=1}^k\lambda_{ia}=t$ is an integer, so we   also have $\sum_{a=1}^k\ell_{ia}=t$.
 \item $M$ contains the singleton sets $\{(ia)\}$, the set $[m]\times [k]$, and the set $[m]\times\{a\}$ for every $a\in [k]$. Note that since $|J_i|=\frac{n}{m}\pm 1=(1\pm\frac{m}{n})\frac{n}{m}$ and $|V_a|=(1\pm\frac{k}{n})\frac{n}{k}$, we have $\sum_{i=1}^m\lambda_{ia}=(1\pm 2\frac{m+k}{n})\frac{mt}{k}$, which means that for $n$ large enough all the sums $\sum_{i=1}^m\ell_{ia}$ will equal $\frac{mt}{k}\pm 1$ (note that $\frac{mt}{k}$ is an integer), and moreover the number of $a\in [k]$ for which $\sum_{i=1}^m\ell_{ia}=\frac{mt}{k}+1$ will equal the number of $a\in [k]$ for which this value is $\frac{mt}{k}-1$.
\end{itemize}
After obtaining the values $\ell_{ia}$ through Lemma \ref{lem:round}, we obtain $\ell'_{ia}$ from $\ell_{ia}$ through the following process: For all $a$ for which $\sum_{i=1}^m\ell_{ia}=\frac{mt}{k}$, we set $\ell'_{ia}=\ell_{ia}$ for all $i\in [m]$. Otherwise we each time take an $a$ for which $\sum_{i=1}^m\ell_{ia}=\frac{mt}{k}+1$ and an $a'$ for which $\sum_{i=1}^m\ell_{ia'}=\frac{mt}{k}-1$. We choose $i$ for which $\ell_{ia}>\ell_{ia'}$ set $\ell'_{ia}=\ell_{ia}-1$, $\ell'_{ia'}=\ell_{ia'}+1$, and for all other $i'$ we set $\ell'_{i'a}=\ell_{i'a}$ and $\ell'_{i'a'}=\ell_{i'a'}$. The resulting $\ell'_{ia}$ satisfy $\sum_{a=1}^k\ell'_{ia}=t$, $\sum_{i=1}^m\ell'_{ia}=\frac{mt}{k}$, and $\ell'_{ia}=\lambda_{ia}\pm 2$.

Now we construct disjoint $K_1,\ldots,K_k\subset [m]\times [t]$ so that for every $i\in [m]$ and every $a\in [k]$ we have $|K_{ia}|=\ell'_{ia}$. By the equations on the sums of $\ell'_{ia}$ above this is doable, and results in $|K_a|=\frac{mt}{k}$ for every $a\in [k]$.

Next, we determine the sizes of the sets $I_{ij}$ of $I'$ and $W_{ijs}$ of $Q''$, through a second use of Lemma \ref{lem:round}. We set the following parameters and multipartitions, for determining $\ell_{ijs}=|W_{ijs}|$.
\begin{itemize}
 \item We plainly set $\lambda_{ijs}=\lambda=\frac{n}{mbt}$ for all $i\in [m]$, $j\in [b]$ and $s\in [t]$. Since all values are the same, the $\ell_{ijs}$ will have value differences bounded by $1$, as befits the equipartition $Q''$.
 \item $M$ consists of the singletons, the set $[m]\times [b]\times [t]$, and the following.
 \begin{itemize}
  \item The set $\bigcup_{s\in K_{ia}}(\{i\}\times \{j\}\times \{s\})$ for every $i\in [m]$, $j\in [b]$ and $a\in [k]$. This will make $U_{ija}$ have size between $\lfloor\frac{|K_{ia}|}{t}|I_{ij}|\rfloor$ and $\lceil\frac{|K_{ia}|}{t}|I_{ij}|\rceil$ (see about $|I_{ij}|$ below).
  \item The set $\{i\}\times \{j\}\times [t]$ for each $i\in [m]$ and $j\in [b]$. Eventually we will have $|I_{ij}|=\sum_{s=1}^t\ell_{ijs}\in\{\lfloor\frac{n}{mb}\rfloor,\lceil\frac{n}{mb}\rceil\}$, so $I'$ will be equitable.
  \item The set $\{i\}\times [b]\times [t]$ for each $i\in [m]$. Eventually we will have $|I_i|=\sum_{s=1}^t\sum_{j=1}^b\ell_{ijs}\in\{\lfloor\frac{n}{m}\rfloor,\lceil\frac{n}{m}\rceil\}$, so $I$ will be equitable.
  \item If $n$ and $m$ are both even, we also add the sets $[1,m/2]\times [b]\times [t]$ and $[m/2+1,m]\times [b]\times [t]$ to $M$. 
  Eventually we will have $\sum_{i=1}^{m/2} |I_i| = \sum_{i=m/2+1}^{m} |I_i| = n/2$.
 \end{itemize}
 \item $N$ consists of the singletons, the set $[m]\times [b]\times [t]$, and the following.
 \begin{itemize}
  \item The set $\{i\}\times [b]\times \{s\}$ for every $i\in [m]$ and $s\in [t]$. This will ensure that the eventual $Q'$ is equitable.
  \item The set $\bigcup_{s\in K_{ia}}(\{i\}\times [b]\times \{s\})$ for every $i\in [m]$ and $a\in [k]$. This will make every $U_{ia}$ have size between $\lfloor\frac{|K_{ia}|}{t}|I_i|\rfloor$ and $\lceil\frac{|K_{ia}|}{t}|I_i|\rceil$.
  \item The set $\bigcup_{(is)\in K_a}(\{i\}\times [b]\times \{s\})$ for every $a\in [k]$. This will ensure that the eventual $Q$ is equitable.
 \end{itemize}
\end{itemize}
After obtaining the values $\ell_{ijs}$ for the respective set sizes $|W_{ijs}|$, we finally construct the partitions themselves. First we construct $I$ as the only interval partition for which $|I_i|=\sum_{s=1}^t\sum_{j=1}^b\ell_{ijs}$ for every $i\in [m]$, and $I'$ as the only refinement of $I$ for which $|I_{ij}|=\sum_{s=1}^t\ell_{ijs}$ for every $i\in [m]$ and $j\in [b]$.
For every $i\in [m]$ and $s\in [t]$ let $b_{is}\in [k]$ be the index such that $(is)\in K_{b_{is}}$. We now go over the indexes $i\in [m]$ and $j\in [b]$, and partition the vertices of $I_{ij}$ into the sets $W_{ijs}$ so that as many members of $V_{b_{is}}\cap I_{ij}$ as possible will go into every $W_{ijs}$. When we can no longer assign vertices in this manner (because $|I_{ij}\cap V_b|$ will not necessarily equal $\sum_{b_{is}=b}\ell_{ijs}$), we assign the remaining vertices to complete the sets that do not yet have the correct size.

Having defined $I$, $I'$ and $Q''$, we define $Q'$ by setting $W_{is}=\bigcup_{j=1}^bW_{ijs}$ for every $i\in [m]$ and $s\in [t]$, and define $Q$ by setting $U_a=\bigcup_{(is)\in K_a}W_{is}$. All properties of $I$, $I'$, $Q$, $Q'$ and $Q''$ immediately follow from the construction, apart from the relationship between $Q$ and $P$ that we still need to prove.

Because of the way we chose the sets $W_{ijs}$ to maximize the number of vertices they contain from the ``correct'' sets of $P$, The partitions $P$ and $Q$ will be $\delta$-close if
$$\sum_{i=1}^m\sum_{j=1}^b\sum_{a=1}^k\big| |V_a\cap I_{ij}|-|U_a\cap I_{ij}| \big|\leq \delta n.$$
Denote the densities according to the string $S_P$ by $d_{P,a}(I_{ij})$ (where $a\in [k]$), and the densities according to $S_Q$ by $d_{Q,a}(I_{ij})$. For $n$ large enough, because $I'$ is an equipartition (interval sizes differ by not more than $1$), we have
$$\frac1n \sum_{i=1}^m\sum_{j=1}^b\sum_{a=1}^k\big| |V_a\cap I_{ij}|-|U_a\cap I_{ij}| \big|\leq \frac{2}{mb}\sum_{i=1}^m\sum_{j=1}^b\sum_{a=1}^k|d_{P,a}(I_{ij})-d_{Q,a}(I_{ij})|.$$
From now on we bound the sums on the right hand side. Recall that $J$ denotes the original interval equipartition of size $m$, and let $J'$ be any refinement of $J$ of size $mb$. 
By Observation \ref{obs:allintapprox}, for $n$ large enough we have
$\frac{1}{mb} \sum_{a \in [k]} \sum_{i \in [m]} \sum_{j \in [b]} |d_{P,a}(I_{ij}) - d_{P,a}(J_{ij})| \leq \delta/20$.
By Lemma \ref{lem:string_robust_to_densities}, we know that $\frac{1}{mb}\sum_{a \in [k]} \sum_{i \in [m]} \sum_{j \in [b]} |d_{P,a}(J_{ij}) - d_{P,a}(J_i)| \leq \delta/20$. 
%$\frac{1}{mb} \sum_{a \in [k]} \sum_{i \in [m]} \sum_{j \in [b]} |d_{Q,a}(I_{ij}) - d_{Q,a}(J_{ij})| \leq \delta/20$,
%and $\frac1m\sum_{a \in [k]} \sum_{i \in [m]} |d_{P,a}(I_i) - d_{P,a}(J_i)| \leq \delta/20$.
Now, recall that we chose the sets $K_a$ so that $|K_{ia}|=t\cdot d_{P,a}(J_i)\pm 2$. This means that $\frac1m\sum_{a \in [k]} \sum_{i \in [m]} |d_{P,a}(J_i) - d_{Q,a}(I_i)| \leq \delta/5$ (recalling also how we chose $t$). Finally, by our construction, for $n$ large enough, $\frac{1}{mb}\sum_{a \in [k]} \sum_{i \in [m]} \sum_{j \in [b]} |d_{Q,a}(I_{i}) - d_{Q,a}(I_{ij})| \leq \delta/20$. This follows from the size restriction that we ensured for the sets $U_{ia}$ and $U_{ija}$.
%This is because we chose $W_{is}=\bigcup_{j=1}^bW_{ijs}$ for every $i\in [m]$ and $j\in [t]$, while $U_r\cap I_i$ is just a union of $W_{is}$ sets for corresponding values of $s$ (which means that all the difference between $d_{Q,a}(I_{ij})$ and $d_{Q,a}(I_i)$ comes from the possible $1$ difference in the relevant set sizes).
Using triangle inequalities with all these bounds on the density differences concludes the proof.
\end{proof}

\begin{lemma}\label{lem:mainpartscheme}
For any positive integer $k$, real value $\gamma$, functions $r:\mathbb{N}\times\mathbb{N}\to\mathbb{N}$ and $f:\mathbb{N}\to\mathbb{N}$, and any $n$-vertex ordered colored graph $G$ (for large enough $n$), there exist an interval equipartition $I$ into $m$ parts where $k\leq m\leq S_{\ref{lem:mainpartscheme}}(\gamma,k,f,r)$, an equipartition $Q'$ of $G$ into $mt$ parts (not necessarily an interval equipartition) which refines $I$ and is additionally $(f,\gamma)$-robust, where $mt\leq T_{\ref{lem:mainpartscheme}}(\gamma,k,f,r)$, and an interval equipartition $I'$ into $m\cdot r(m,t)$ parts also refining $I$, so that the LCR $Q''=Q'\sqcap I'$ is an equipartition into exactly $mt\cdot r(m,t)$ parts (so each set of $Q'$ intersects ``nicely'' all relevant intervals in $I'$).

Moreover, if $n$ is even, then $m$ will be even and $I$ will respect the middle.
\end{lemma}

\begin{proof}
For each $l\in\mathbb{N}$ we define a function $g_l:\mathbb{N}\to\mathbb{N}$ by setting for every $m\in\mathbb{N}$
$$g_l(m)=m\cdot r(m,T_{\ref{lem:main_equiparts}}(\delta_{\ref{lem:approx_robust1}}(\gamma),l,m)/m).$$
Then we define a function $h:\mathbb{N}\to\mathbb{N}$ setting for every $l\in\mathbb{N}$
$$h(l)=f(T_{\ref{lem:main_equiparts}}(\delta_{\ref{lem:approx_robust1}}(\gamma),l,T_{\ref{obs:robust_partition_intervals}}(k,g_l,\gamma_{\ref{lem:main_equiparts}}(\delta_{\ref{lem:approx_robust1}}(\gamma),l)))).$$

We start with an equipartition $P$ that is $(h,\delta_{\ref{lem:approx_robust1}}(\gamma))$-robust that 
we obtain by Observation \ref{obs:robust_partition_colored_graphs}, and then with respect to the string $S_P$ we obtain by Observation \ref{obs:robust_partition_intervals} an interval equipartition $J$ that has at least $k$ parts and is $(g_{|P|},\gamma_{\ref{lem:main_equiparts}}(\delta_{\ref{lem:approx_robust1}}(\gamma),|P|))$-robust. Note that $|P|\leq T_{\ref{obs:robust_partition_colored_graphs}}(1,h,\delta_{\ref{lem:approx_robust1}}(\gamma))$, and hence $|J|\leq T_{\ref{obs:robust_partition_intervals}}(k,g_{T_{\ref{obs:robust_partition_colored_graphs}}(1,h,\delta_{\ref{lem:approx_robust1}}(\gamma))}\,,\gamma_{\ref{lem:main_equiparts}}(\delta_{\ref{lem:approx_robust1}}(\gamma),T_{\ref{obs:robust_partition_colored_graphs}}(1,h,\delta_{\ref{lem:approx_robust1}}(\gamma))))$, which we set as our $S_{\ref{lem:mainpartscheme}}(\gamma,k,f,r)$.

If $n$ is even, then we make sure that $k$ is also even (otherwise we replace it with $k+1$ in all of the above), and then $|J|$ will be even as well (and our subsequent use of Lemma \ref{lem:main_equiparts} will provide an $I$ that respects the middle).

We then invoke Lemma \ref{lem:main_equiparts} to get our partitions $I=I_{\ref{lem:main_equiparts}}(\delta_{\ref{lem:approx_robust1}}(\gamma),P,|J|,g_{|P|}(|J|)/|J|)$, $I'=I'_{\ref{lem:main_equiparts}}(\delta_{\ref{lem:approx_robust1}}(\gamma),P,|J|,g_{|P|}(|J|)/|J|)$, and $Q' = Q'_{\ref{lem:main_equiparts}}(\delta_{\ref{lem:approx_robust1}}(\gamma),P,|J|,g_{|P|}(|J|)/|J|)$. By the size guarantees of Lemma \ref{lem:main_equiparts} we have $|I|=|J|$ (ensuring our size bound for $|I|$), and $|Q'|$ is bounded by $T_{\ref{lem:main_equiparts}}(\delta_{\ref{lem:approx_robust1}}(\gamma),T_{\ref{obs:robust_partition_colored_graphs}}(1, h,\delta_{\ref{lem:approx_robust1}}(\gamma)),S_{\ref{lem:mainpartscheme}}(\gamma,k,f,r))$, which we set as our $T_{\ref{lem:mainpartscheme}}(\gamma,k,f,r)$.

Lemma \ref{lem:main_equiparts} guarantees all requirements apart from the robustness of $Q'$. To prove it, we note that $Q'$ is a refinement of the partition $Q=Q_{\ref{lem:main_equiparts}}(\delta_{\ref{lem:approx_robust1}}(\gamma),P,|J|,g_{|P|}(|J|)/|J|)$ into at most $T_{\ref{lem:main_equiparts}}(\delta_{\ref{lem:approx_robust1}}(\gamma),|P|,T_{\ref{obs:robust_partition_intervals}}(k,g_{|P|},\gamma_{\ref{lem:main_equiparts}}(\delta_{\ref{lem:approx_robust1}}(\gamma),|P|)))$ parts, where $|Q|=|P|$ and $Q$ and $P$ are $\delta_{\ref{lem:approx_robust1}}(\gamma)$-close. Hence by invoking Lemma \ref{lem:approx_robust1} (which makes $Q$ $(h,\gamma)$-robust), and then Observation \ref{obs:chainrobust}, we get that $Q'$ is indeed $(f,\gamma)$-robust.
\end{proof}

%%%%%%%%%%%%%%%%%%%%%%%%%%%%%%%%%%%%%%%%%%%%%%%%%%%%%%%%%
\section{The finite case for graphs}
%%%%%%%%%%%%%%%%%%%%%%%%%%%%%%%%%%%%%%%%%%%%%%%%%%%%%%%%%
\label{sec:proof_main_result}
This section contains the proof of Theorem \ref{thm:infinite_removal_lemma} for the case that the forbidden family $\mathcal{F}$ is finite.
This is the ordered generalization of the finite induced graph removal lemma (Theorem \ref{thm:induced_graph_removal_lemma_AFKS}).
\begin{thm}[Finite ordered graph removal lemma]
	\label{thm:finite_removal_lemma}
	Fix a finite set $\Sigma$ with $|\Sigma| \geq 2$.
	For any finite family $\mathcal{F}$ of ordered graphs $F:\binom{[n_F]}{2} \to \Sigma$
	and any $\epsilon > 0$ there exists
	$\delta = \delta(\mathcal{F}, \epsilon) > 0$, such that any ordered graph $G:\binom{V}{2} \to \Sigma$ that is
	$\epsilon$-far from $\mathcal{F}$-freeness contains at least $\delta n^q$ induced copies of some graph $F \in \mathcal{F}$.
\end{thm}
The proof of Theorem \ref{thm:infinite_removal_lemma} is completed in Section \ref{sec:proofinfinite}, by considering the case where $\mathcal{F}$ is infinite. The proof for the infinite case mostly relies on ideas and tools presented in this section, but requires another step, which is motivated by the ideas of Alon and Shapira \cite{AlonShapira2008} for the unordered case.

\subsection{Representing subsets}
Fix a finite alphabet $\Sigma$ and a finite family $\mathcal{F}$ over $\Sigma$. 
Let $d_{\mathcal{F}}$ denote the largest number of vertices in a graph from $\mathcal{F}$.
%and define the function $r:\mathbb{N} \times \mathbb{N} \to \mathbb{N}$ as follows:
%$r(m, t) = R_{\ref{coro:ramsey_app}}(2^{|\Sigma| t^2} , m, d_{\mathcal{F}})$ for any $m,t \in \mathbb{N}$.
%Also pick $k = \lceil 20 / \epsilon \rceil$. 
Now let $G = (V, c)$ be an $n$-vertex $\Sigma$-colored graph 
and suppose that $I, I', Q', Q''$ are equipartitions of $G$ of sizes $m, mb, mt, mbt$ respectively, so that %$m \geq k$, 
$I, I'$ are interval partitions,
$I'$ and $Q'$ refine $I$, and $Q'' = I' \sqcap Q'$. % and $b = r(m, t)$. 
More specifically, we write
$I =  (I_1, \ldots, I_m)$, 
$I' = (I_{11}, \ldots, I_{1b}, \ldots, I_{m1}, \ldots, I_{mb})$, 
$Q' = (U_{11}, \ldots, U_{1t}, \ldots, U_{m1}, \ldots, U_{mt})$, 
$Q'' = (U_{111}, \ldots, U_{mbt})$, where $I_j = \bigcup_{r=1}^{b} I_{jr} = \bigcup_{s=1}^{t} U_{js}$ for any $j \in [m]$ and $U_{jrs} = I_{jr} \cap U_{js}$ for any $j \in [m], r \in [b], s \in [t]$.
Note that this is the same setting as the one obtained in Lemma \ref{lem:mainpartscheme}, but we do not apply the lemma at this point; in particular,
we currently do not make any assumptions on the equipartitions other than those stated above.
%first we will present several lemmas that will allow us to determine the desired parameters with which we will use Lemma \ref{lem:mainpartscheme}. 
We may and will assume whenever needed that $n$ is large enough (as a function of
all relevant parameters), and that any tuple of subsets of $V$ considered in this section has at least two parts (i.e., it is not trivial).

\begin{definition}[Representing subsets]
Let $\alpha, \beta, \mu > 0$ be real numbers
and suppose that $A = (A_1, \ldots, A_l)$ is an equipartition of $G$.
We say that $B = (B_1, \ldots, B_l)$ \emph{represents} $A$ if $B_i \subseteq A_i$ for any $i \in [l]$.
Furthermore, we say that $B$ \emph{$(\alpha, \beta, \mu)$-represents} $A$ if the following holds. 
\begin{itemize}
\item $B_i \subseteq A_i$ and $|B_i| \geq  \alpha n$ for any $i \in [l]$.
\item All pairs $(B_i, B_j)$ with $i<j \in [l]$ are $\beta$-regular.
\item $\frac{1}{\binom{l}{2}}\sum_{i < j \in [l]} \sum_{\sigma \in \Sigma} |d_{\sigma}(B_i, B_j) - d_{\sigma}(A_i, A_j)| \leq \mu$.
\end{itemize}
\end{definition}

The following lemma is a slight variant of Corollary 3.4 in \cite{AlonFKS00}, suggesting that partitions that are robust
enough have good representing subsets. 
The proof follows along the same lines of the proof of Lemma 3.2 in \cite{ConlonFox2013}, so we omit it.
\begin{lemma}[\cite{AlonFKS00, ConlonFox2013}]
\label{lem:representatives1}
For any $\mu > 0$ and function $\beta:\mathbb{N} \to (0, 1)$ there exist a function $f = f^{(\beta, \mu)}_{\ref{lem:representatives1}}:\mathbb{N} \to \mathbb{N}$ and a real number $\gamma = \gamma_{\ref{lem:representatives1}}(\mu) > 0$, such that
for any integer $l > 0$  
there is a real number $\alpha = \alpha_{\ref{lem:representatives1}}(\beta, \mu, l) > 0$, 
all satisfying the following.
If $A = (A_1, \ldots, A_l)$ is an $(f, \gamma)$-robust equipartition of $G$, then
there exists a tuple $B = (B_1, \ldots, B_l)$ which $(\alpha, \beta(l), \mu)$-represents $A$.
\end{lemma}

The next lemma is not hard to derive from Lemma \ref{lem:representatives1} using Lemma \ref{lem:robust_to_densities}, and is more suitable to our setting.
\begin{lemma}
\label{lem:double_representing}
For any function $\beta:\mathbb{N} \to (0, 1)$,  function $g:\mathbb{N} \to \mathbb{N}$, and real number $\mu > 0$, 
there exist a function $f = f_{\ref{lem:double_representing}}^{\beta, g, \mu}:\mathbb{N} \to \mathbb{N}$ and a real number $\gamma = \gamma_{\ref{lem:double_representing}}(\mu) > 0$, %(g, \beta, \mu) > 0$, 
so that for any integer $l > 1$ there exists
$\alpha = \alpha_{\ref{lem:double_representing}}(\beta, g, \mu, l) > 0$ satisfying the following:
If $A = (A_1, \ldots, A_l)$ is an $(f, \gamma)$-robust equipartition of $G$
and $A' = (A_{11}, \ldots, A_{1L}, \ldots, A_{l1}, \ldots A_{lL})$ is an equitable refinement of $A$,
where $lL \leq g(l)$ and  $A_i = \bigcup_{j=1}^{L} A_{ij}$ for any $i \in [l]$, 
then there exists $B = (B_{11}, \ldots, B_{1L}, \ldots B_{l1} ,\ldots, B_{lL})$ which $(\alpha, \beta(lL), \mu)$-represents $A'$,
and satisfies 
\[
\frac{1}{\binom{l}{2} L^2} \sum_{i<i' \in [l]} \sum_{j,j' \in [L]} \sum_{\sigma \in \Sigma} |d_{\sigma}(B_{ij}, B_{i'j'}) - d_{\sigma}(A_i, A_{i'})| \leq 2\mu.
\]
\end{lemma}
\begin{proof}
Pick $f = f_{\ref{lem:double_representing}}^{\beta, g, \mu} = 
f_{\ref{lem:representatives1}}^{(\beta, \mu)} \circ g$
and $\gamma = \gamma_{\ref{lem:double_representing}}(\mu) = \min\{\delta_{\ref{lem:robust_to_densities}}(|\Sigma|, \mu), \gamma_{\ref{lem:representatives1}}(\mu)\}$.
Also pick $\alpha = \alpha_{\ref{lem:double_representing}}(\beta, g, \mu, l) = \alpha_{\ref{lem:representatives1}}(\beta, \mu, g(l))$, 
and suppose that $A$ is $(f, \gamma)$-robust.
By Observation \ref{obs:chainrobust} and the fact that $|A'| = lL \leq g(l)$, we know that $A'$ is $(f_{\ref{lem:representatives1}}^{(\beta, \mu)}, \gamma_{\ref{lem:representatives1}}(\mu))$-robust, so by Lemma \ref{lem:representatives1} there exists a tuple
$B = (B_{11}, \ldots, B_{1L}, \ldots B_{l1} ,\ldots, B_{lL})$ which $(\alpha_{\ref{lem:representatives1}}(\beta, \mu, lL), \beta(lL), \mu)$-represents $A'$, and by the monotonicity of $\alpha$, $B$ also $(\alpha, \beta(lL), \mu)$-represents $A'$. In particular, 
\[
\frac{1}{\binom{l}{2} L^2} \sum_{i < i' \in [l]} \sum_{j,j' \in [L]} \sum_{\sigma \in \Sigma} 
|d_{\sigma}(B_{ij}, B_{i'j'}) - d_{\sigma}(A'_{ij}, A'_{i'j'})| \leq \mu.
\]
Now by Lemma \ref{lem:robust_to_densities}, and since $|A'| \leq g(l) \leq f(l)$,
\[
\frac{1}{\binom{l}{2} L^2} \sum_{i < i' \in [l]} \sum_{j,j' \in [L]} \sum_{\sigma \in \Sigma} 
|d_{\sigma}(A'_{ij}, A'_{i'j'}) - d_{\sigma}(A_{i}, A_{i'})| \leq \mu.
\]
Combining the above two inequalities and using the triangle inequality concludes the proof.
\end{proof}

\subsection{The graph of the representatives and its coloring}
For the next step, let $\Gamma = \Gamma(\Sigma, t)$ denote the collection of all $t \times t$ matrices $M$ of the following form: Each entry of $M$ is a non-empty subset of the color set $\Sigma$ (where a subset is allowed to appear in multiple entries of $M$), so $|\Gamma(\Sigma, t)| < 2^{|\Sigma|t^2}$.

\begin{definition}[Threshold color matrices, threshold graphs, undesirability] 
Suppose \\
that $W = (W_{111}, \ldots, W_{mbt})$ represents $Q''$ and define $W_{jr} = (W_{jr1}, \ldots, W_{jrt})$ and $X_{j} = (U_{j1}, \ldots, U_{jt})$ for any $j \in [m]$ and $r \in [b]$. Let $0 < \eta < \rho < 1/|\Sigma|$ be real numbers.

For two $t$-tuples $A = (A_1, \ldots, A_t)$ and $B = (B_1, \ldots, B_t)$ where $A_s, B_s \subseteq V$ for any $s \in [t]$, the \emph{$\eta$-threshold matrix} $M = M(A, B, \eta) \in \Gamma$ of the pair $A, B$ is the $t \times t$ matrix whose $(s,s')$ entry (for $(s,s') \in [t]^2$) is the set of colors $\sigma \in \Sigma$ that satisfy $d_{\sigma}(A_s, B_{s'}, c\restrict{A_s, B_{s'}}) \geq \eta$. Note that this set cannot be empty since $\eta < 1 / |\Sigma|$.

The \emph{$(\eta, W)$-threshold graph} $H_W^{\eta}$ is an (ordered) $\Gamma$-colored graph defined as follows: The vertices of $H_W^{\eta}$ are all parts of $I'$,
and the color of the edge $I_{jr} I_{j'r'}$ is $M(W_{jr}, W_{j'r'}, \eta)$.

The edge $I_{jr}I_{j'r'}$ of $H_W^{\eta}$ is $\rho$-undesirable if $j' > j$ and at least $\rho t^2$ of the pairs $(s,s') \in [t]^2$ satisfy $M(X_j, X_{j'}, \rho)[s,s'] \nsubseteq M(W_{jr}, W_{j'r'}, \eta)[s,s']$.
Finally, $H_W^{\eta}$ is $\rho$-undesirable if at least $\rho \binom{m}{2} b^2$ of the edges $I_{jr} I_{j'r'}$ in it are $\rho$-undesirable,
and $\rho$-desirable otherwise.
\end{definition}
In other words, an edge $I_{jr}I_{j'r'}$ is undesirable if there are many pairs of sets $W_{jrs}, W_{j'r's'}$ in $W$, for which the density of some original edge color in $W_{jrs} \times W_{j'r's'}$ is significantly smaller than its density in $U_{js} \times U_{j's'}$.
$H_W^{\eta}$ is undesirable if it contains many undesirable edges.
Note that the set of $\rho$-undesirable edges in $H_W^{\eta}$ is orderly: Whether an edge $I_{jr}, I_{j'r'}$ of $H_W^{\eta}$ 
is undesirable or not depends only on its color $M(W_{jr}, W_{j'r'}, \eta)$ and on $M(X_j, X_{j'}, \rho)$.

The following lemma relates the robustness of our partitions to the desirability of the resulting threshold charts.
\begin{lemma}
\label{lem:robust_undesirable}
For any $0 < \rho  < 1/|\Sigma|$ and functions $\beta:\mathbb{N} \to (0,1/|\Sigma|)$ and $g:\mathbb{N} \to \mathbb{N}$, there exist a function 
$f = f_{\ref{lem:robust_undesirable}}^{\rho, \beta, g}:\mathbb{N} \to \mathbb{N}$ and positive real numbers  $\mu = \mu_{\ref{lem:robust_undesirable}}(\rho) \leq \rho$,  $\gamma = \gamma_{\ref{lem:robust_undesirable}}(\rho)$ and $\alpha = \alpha_{\ref{lem:robust_undesirable}}(\rho, \beta, g, m, t)$, 
such that if $Q'$ is $(f, \gamma)$-robust and $|Q''| \leq g(|Q'|)$, then there is a tuple $W = (W_{111}, \ldots, W_{mbt})$ which $(\alpha, \beta(mbt), \mu)$-represents $Q''$, and furthermore $H_W^{\rho/2}$ is $\rho$-desirable.

\end{lemma}
\begin{proof}
Let $0 < \rho < 1/|\Sigma|$ and suppose that $H_{W}^{\rho/2}$ is $\rho$-undesirable, where $W$ is any tuple that represents $Q''$. The definition of undesirability implies that
\begin{align}
\label{eq:densities_dont_differ}
\frac{1}{\binom{m}{2} t^2 b^2} \sum_{j<j' \in [m]} \sum_{s,s' \in [t]} \sum_{r,r' \in [b]} \sum_{\sigma \in \Sigma} |d_{\sigma}(W_{jrs}, W_{j'r's'}) - d_{\sigma}(U_{js}, U_{j's'})| \geq 
\frac{ \rho \binom{m}{2} b^2 \rho t^2 \rho/2}{\binom{m}{2} t^2 b^2}= \frac{\rho^3}{2}.
\end{align}
Indeed, if $M(X_j, X_{j'}, \rho)[s,s'] \nsubseteq M(W_{jr}, W_{j'r'}, \rho/2)[s,s']$
then there exists some $\sigma \in \Sigma$ for which $d_{\sigma}(U_{js}, U_{j's'}) \geq \rho$ but $d_{\sigma}(W_{jrs}, U_{j'r's'}) \leq \rho/2$,
so each such event contributes $\rho/2$ to the sum in the left hand side.

Therefore, $H_W^{\rho/2}$ is $\rho$-desirable if the above sum is smaller than $\rho^3/2$.
Thus, we pick $\mu(\rho) = \rho^3 / 5$. %and take the function $g(l) = lr(l, l)$. Note that $g(|Q'|) = g(mt) = m t r(mt,mt) \geq mtr(m,t) = mt b = |Q''|$.
Also pick $f_{\ref{lem:robust_undesirable}}^{\rho, \beta, g} = f_{\ref{lem:double_representing}}^{\beta, g, \mu}$, $\gamma_{\ref{lem:robust_undesirable}}(\rho) = \gamma_{\ref{lem:double_representing}}(\mu)$, and $\alpha_{\ref{lem:robust_undesirable}}(\rho, \beta, g, m, t) = \alpha_{\ref{lem:double_representing}}(\beta, g, \mu, mt)$.
Since $Q'$ is $\left(f_{\ref{lem:double_representing}}^{\beta, g, \mu}, \gamma_{\ref{lem:double_representing}}(\mu)\right)$-robust, and since $|Q''| \leq g(|Q'|)$, Lemma \ref{lem:double_representing} 
implies that there exists $W = (W_{111}, \ldots, W_{mbt})$ which $(\alpha, \beta(mbt), \mu)$-represents $Q'$, also guaranteeing that the left hand side of \eqref{eq:densities_dont_differ} is at most $2\mu < \rho^3/2$, so $H_W^{\rho/2}$ is $\rho$-desirable.
\end{proof}

\begin{definition}[Nicely colored subgraph]
\label{def:nicely_colored_subgraph}
Let $W = (W_{111}, \ldots, W_{mbt})$ be a tuple of subsets that represents $Q''$ and 
let $\eta > 0$. A subgraph $D = (\bigcup_{j=1}^{m} D_j, c_D)$ of $H_W^{\eta}$ is said to be \emph{nicely colored} if the following conditions hold.
\begin{itemize}
\item For any $j \in [m]$, $D_j \subseteq I_j$ and $|D_j| = d_{\mathcal{F}}$.
\item For any fixed $j \in [m]$, all edges inside $D_j$ have the same color from $\Gamma$, denoted by $C_{jj}^{(D)}$.
\item For any fixed $j < j' \in [m]$, all edges between $D_j$ and $D_{j'}$ have the same color from $\Gamma$, denoted by $C_{jj'}^{(D)}$.
\end{itemize}
\end{definition}

The next lemma follows directly from Corollary \ref{coro:ramsey_app}. %on our above choices of $r$ and $b = r(m,t)$. Recall that the choice of $r$ depends only on the alphabet $\Sigma$ and the family $\mathcal{F}$.
\begin{lemma}
\label{lem:representing_ramsey}
For any two integers $m,t > 0$ there exists $R = R_{\ref{lem:representing_ramsey}}(m, t)$ satisfying the following: If $b \geq R_{\ref{lem:representing_ramsey}}(m, t)$, 
then for any tuple $W = (W_{111}, \ldots, W_{mbt})$ that represents $Q''$ and any
$\eta > 0$ there exists a nicely colored subgraph $D = D_{\ref{lem:representing_ramsey}}(W, \eta)$ of $H_W^{\eta}$.
Moreover, if $H_W^{\eta}$ is $\rho$-desirable for some $\eta < \rho < 1 / |\Sigma|$, then the number of $\rho$-undesirable edges in $D$ is at most $2 \rho \binom{m}{2} (d_{\mathcal{F}})^2$.
\end{lemma}
\begin{proof}
Take $R_{\ref{lem:representing_ramsey}}(m, t) = R_{\ref{coro:ramsey_app}}(2^{|\Sigma| t^2} , m, d_{\mathcal{F}}) > R_{\ref{coro:ramsey_app}}(|\Gamma_t| , m, d_{\mathcal{F}})$. 
Since the set of $\rho$-undesirable edges in $H_W^{\eta}$ is orderly,
 we may apply Corollary \ref{coro:ramsey_app} on $H_W^{\eta}$, to get a nicely colored subgraph $D$ of it.
If $H_W^{\eta}$ is $\rho$-desirable for some $\eta < \rho < 1 / |\Sigma|$, then by definition it has at most $\rho \binom{m}{2} b^2$ $\rho$-undesirable edges, and so the last condition in Corollary \ref{coro:ramsey_app} implies that $D$ has at most $2 \rho \binom{m}{2} (d_{\mathcal{F}})^2$ $\rho$-undesirable edges.
\end{proof}

%Suppose now that $W = (W_{111}, \ldots, W_{mbt})$ is a tuple of subsets of the respective parts in $Q''$ and consider the graph $H_W^{\eta}$ for some $\eta > 0$. 
%%$\Gamma$-colored graph $H_W^{\rho/2}$ is $\rho$-desirable (the values of $\alpha, \lambda, \rho$ may be arbitrary for now).
%Our choice of the function $r$ and the fact that $b = r(m, t)$ imply the existence of a subgraph $D \subseteq H_W^{\rho/2}$ whose vertex set is $\bigcup_{j=1}^{m} D_j$, where $D_j \subseteq I_j$ is of size $d = d_{\mathcal{F}}$, all edges inside each $D_j$ have the same color, denoted $M_{jj}$, all edges in $D_j \times D_{j'}$ have the same color, denoted by $M_{jj'}$ and additionally, if $H_W^{\eta}$ is $\rho$-desirable, then the number of undesirable edges in $D$ is at most $2 \rho \binom{m}{2} d^2$.

\subsection{Cleaning the original graph}
\begin{definition}[Cleaned graph]
\label{def:cleaned_graph}
Let $W = (W_{111}, \ldots, W_{mbt})$ be a tuple of subsets which represents $Q''$, 
let $\eta > 0$, and suppose that $D$ is a nicely colored subgraph of $H_W^{\eta}$.
%For any $j < j'$ let $C_{jj'} \in \Gamma$ denote the common color of all edges between $D \cap {I_j}$ and $D \cap I_{j'}$; also let $C_{jj} \in \Gamma$ denote the common color of all edges inside $D \cap I_j$.
 The \emph{cleaned graph} $G' = G'(G, D) = (V, c')$ is defined as follows.
For any $u < v \in V$ where $u \in I_{js}$ and $v \in I_{j's'}$, we set $c'(uv) = c(uv)$ if $c(uv) \in C_{jj'}^{(D)}[s,s']$, and otherwise we set $c'(uv)$ to an arbitrary color from $C_{jj'}^{(D)}[s,s']$.
\end{definition}

The next lemma states that if $D$ comes from a desirable $H_W^{\eta}$, then $G'(G, D)$ is close to $G$.
\begin{lemma}
\label{lem:easy_to_clean}
Suppose that $D$ is a nicely colored subgraph
of some $H_W^{\eta}$ with $W$ representing $Q''$ and $0 < \eta < \rho$, such that at most $2 \rho \binom{m}{2} d_{\mathcal{F}}^2$ edges of $D$ are $\rho$-undesirable.
Then $G' = G'(G, D)$ is $(7 |\Sigma| \rho + 2/m)$-close to $G$, where $m = |I|$.
\end{lemma}
\begin{proof}
Write $G' = (V, c')$ and let $\mathcal{J}$ denote the set of pairs $j < j' \in [m]$
such that $D_j \times D_{j'}$ contains an undesirable edge.
An edge $e \in \binom{V}{2}$ may satisfy $c'(e) \neq c(e)$ only if at least one of the following holds (some of the inequalities stated below rely on the assumption that $n$ is large enough).
\begin{enumerate}
\item $e$ lies inside some part $I_j$ of $I$. The number of such edges is $\sum_{j=1}^{m} \binom{|I_j|}{2} \leq m \binom{\lceil n/m \rceil}{2} < \frac{2}{m} \binom{n}{2}$.

\item $e \in I_{j_1} \times I_{j_2}$ where $(j_1, j_2) \in \mathcal{J}
$. But $|\mathcal{J}| \leq 2\rho \binom{m}{2}$: The number of $\rho$-undesirable edges in $D$ is exactly $|\mathcal{J}| d_{\mathcal{F}}^2$,
since $D$ is orderly (with respect to the parts $D_1, \ldots, D_m$) and has $d_{\mathcal{F}}$ vertices in each $D_i$.
Thus, $|\mathcal{J}| d_{\mathcal{F}}^2 \leq 2 \rho \binom{m}{2} d_{\mathcal{F}}^2$, which implies the desired inequality. 
Therefore, the number of edges $e$ of of this type is less than $3 \rho \binom{n}{2}$.

\item $e \in U_{js} \times U_{j's'}$ where $j < j' \in [m]$, $(j,j') \notin \mathcal{J}$, 
and $M(X_j, X_{j'}, \rho)[s,s'] \nsubseteq C_{jj'}(D)[s,s']$.
But since the number of pairs $(s,s') \in [t]^2$ that satisfy this condition for a fixed $(j,j') \notin \mathcal{J}$ is at most $\rho t^2$, only at most $ 3 \rho |I_j| |I_{j'}| / 2$ of the edges $e \in I_j \times I_{j'}$ belong here, implying that the total number of edges of this type is less than $2 \rho \binom{n}{2}$.

\item $e \in U_{js} \times U_{j's'}$ where $j < j' \in [m]$, $(j,j') \notin \mathcal{J}$, 
and $M(X_j, X_{j'}, \rho)[s,s'] \subseteq C_{jj'}(D)[s,s']$, but $d_{c(e)}(U_{js}, U_{j's'}) < \rho$. The number of such edges in $U_{js} \times U_{j's'}$ is at most $ |\Sigma| \cdot \rho |U_{js}| |U_{j's'}|$, and the total number of such edges is less than $2 \rho |\Sigma| \binom{n}{2}$.
\end{enumerate}
Therefore, the total number of edges $e$ with $c(e) \neq c'(e)$ is less than $(7 \rho |\Sigma| + 2/m) \binom{n}{2}$.
\end{proof}

\begin{lemma}
\label{lem:has_good_copy}
Let $W = (W_{111}, \ldots, W_{mbt})$ be a tuple that represents $Q''$ and let $\eta > 0$. 
If $D$ is a nicely colored subgraph of $H_W^{\eta}$ and the cleaned $G'(G, D)$ contains a copy of some $F = ([n_F], c_F) \in \mathcal{F}$, then there exist $W_{j_1 r_1 s_1}, \ldots, W_{j_{n_F} r_{n_F} s_{n_F}} \in W$ with the following properties.
\begin{itemize}
\item For any $i \in [n_F-1]$, either $j_{i+1} > j_i$, or $j_{i+1} = j_i$ and $r_{i+1} > r_i$.
\item For any $i < i' \in [n_F]$ it holds that $d_{c_F(ii')}(W_{j_i r_i s_i}, W_{j_{i'} r_{i'} s_{i'}}) \geq \eta$.
\end{itemize}
\end{lemma}
\begin{proof}
Suppose that $G'(G,D) = (V,c')$ contains a copy of $F$ whose vertices in $V$ are $v_1 < \ldots < v_{n_F}$. 
For any $i \in [n_F]$, let $j_i \in [m], s_i \in [t]$ be the indices for which 
$v_i \in I_{j_i s_i}$ and denote the vertices of $D$ inside $I_{j_i}$ by $D_{j_i} = \{I_{j_i r_{i1}}, \ldots, I_{j_i r_{id_{\mathcal{F}}}}\}$, where $r_{i1} < \ldots < r_{i d_{\mathcal{F}}} \in [b]$ for any $ i \in [n_F]$.
Then for any $i,i' \in [n_F]$ and $l,l' \in [d_{\mathcal{F}}]$, for which either $i < i'$, or $i = i'$ and $l < l'$,
it holds that $c_F(ii')  = c'(v_i, v_{i'}) \in  C_{j_i j_{i'}}^{(D)}[s_i,s_{i'}] = M(W_{j_i r_{il}}, W_{j_{i'} r_{i'l'}}, \eta)[s_i,s_{i'}]$, and so by definition
$d_{c_F(ii')}(W_{j_i r_{il} s_{i}}, W_{j_{i'} r_{i'l'} s_{i'}}) \geq \eta$.

Therefore, the sets $W_{j_1 r_{11} s_1}, \ldots, W_{j_{n_F} r_{n_F n_F} s_{n_F}}$ satisfy the conditions of the lemma: 
They exist, since $n_F \leq d_{\mathcal{F}}$. The first condition holds since $j_1 \leq \ldots \leq j_{n_F}$, and if $j_{i} = j_{i+1}$ then $r_{ii} = r_{(i+1)i} < r_{(i+1)(i+1)}$. 
The second condition holds by the first paragraph of the proof (putting $l = i$ and $l' = i'$). 
\end{proof}

\subsection{Proof of Theorem \ref{thm:finite_removal_lemma}}
Suppose that $G$ is $\epsilon$-far from $\mathcal{F}$-freeness.
Take the function $r = R_{\ref{lem:representing_ramsey}}$ (note that $r$ is a two-variable function)
and let $g:\mathbb{N} \to \mathbb{N}$ be defined by $g(l) = lr(l, l)$ for any $l \in \mathbb{N}$. 
Also take $k = \lceil 20 / \epsilon \rceil$,
$\rho = \epsilon / 8 |\Sigma|$,  and $\beta:\mathbb{N} \to (0,1/|\Sigma|)$ as a constant function that satisfies $\beta(l) = \beta_{\ref{lem:regularity_to_many_copies}}^{\rho/2}(d_{\mathcal{F}})$ for any $l \in \mathbb{N}$.
Also take 
$f = f_{\ref{lem:robust_undesirable}}^{\rho, \beta, g}$, and $\gamma = \gamma_{\ref{lem:robust_undesirable}}(\rho)$.

Apply Lemma \ref{lem:mainpartscheme} with parameters $k, \gamma, r, f$, obtaining the equipartitions $I, I', Q', Q''$ of sizes $m, mb, mt, mbt$ as in the statement of the lemma, where $k \leq m \leq S_{\ref{lem:mainpartscheme}}(\gamma,k,f,r)$, $mt \leq T_{\ref{lem:mainpartscheme}}(\gamma,k,f,r)$, $b = r(m,t) = R_{\ref{lem:representing_ramsey}}(m, t)$, and $Q'$ is $(f, \gamma)$-robust.
Observe that $|Q''| = mt r(m, t) \leq g(mt) = g(|Q'|)$.

Next, define $\alpha = \alpha_{\ref{lem:robust_undesirable}}(\rho, \beta, g, S_{\ref{lem:mainpartscheme}}(\gamma, k, f, r), T_{\ref{lem:mainpartscheme}}(\gamma, k, f, r))$ and $\mu = \mu_{\ref{lem:robust_undesirable}}(\rho)$.
By Lemma \ref{lem:robust_undesirable}, and since $\beta(l) = \beta_{\ref{lem:regularity_to_many_copies}}^{\rho/2}(d_{\mathcal{F}})$ for any $l 
\in \mathbb{N}$,
there is a tuple $W$ which $(\alpha, \beta_{\ref{lem:regularity_to_many_copies}}^{\rho/2}(d_{\mathcal{F}}), \mu)$-represents $Q''$, and $H_W^{\rho/2}$ is $\rho$-desirable.
By Lemma \ref{lem:representing_ramsey}, and since $b = R_{\ref{lem:representing_ramsey}}(m,t)$, there is a nicely colored subgraph $D = D_{\ref{lem:representing_ramsey}}(W, \rho/2)$, containing at 
most $2\rho \binom{m}{2} (d_{\mathcal{F}})^2$ $\rho$-undesirable edges.

Lemma \ref{lem:easy_to_clean} implies that $G' = G'(G, D)$ is $(7|\Sigma|\rho + 2/m)$-close to $G$; but $7|\Sigma|\rho + 2/m \leq 7\epsilon/8 + 2/k < \epsilon$,
so $G'$ contains a copy of some $F = ([n_F], c_F) \in \mathcal{F}$. 
Therefore, by Lemma \ref{lem:has_good_copy} (putting $\eta = \rho/2$ in the statement of the lemma),  there exist $W_{j_1 r_1 s_1}, \ldots, W_{j_{n_F}, r_{n_F}, s_{n_F}} \in W$ that satisfy the conditions of the lemma. As all pairs of sets from $W$ are $\beta_{\ref{lem:regularity_to_many_copies}}^{\rho/2}(n_F)$-regular (since $n_F \leq d_{\mathcal{F}}$), we can apply Lemma \ref{lem:regularity_to_many_copies} to conclude that the number of $F$-copies in $G$ is at least 
$\delta n^q$ for $q = n_F \leq d_{\mathcal{F}}$ and
$\delta = \kappa_{\ref{lem:regularity_to_many_copies}}(\rho/2, n_F) \alpha^{n_F} \geq \kappa_{\ref{lem:regularity_to_many_copies}}(\rho/2, d_{\mathcal{F}}) \alpha^{d_{\mathcal{F}}}$, concluding the proof.

%%%%%%%%%%%%%%%%%%%%%%%%%%%%%%%%%%%%%%%%%%%%%%%%%%%%%%%%
\section{The infinite case}\label{sec:proofinfinite}
In this section we use the same notation as in Section \ref{sec:proof_main_result}, unless stated otherwise.
The proof of Theorem \ref{thm:infinite_removal_lemma} follows that of Theorem \ref{thm:finite_removal_lemma} almost word by word, with only one major difference:
In the proof of Theorem \ref{thm:finite_removal_lemma} we have picked $d_{\mathcal{F}}$
to be the largest number of vertices in a graph from $\mathcal{F}$, and showed that 
if $G$ is $\epsilon$-far from $\mathcal{F}$-freeness than there must be 
a set of at most $d_{\mathcal{F}}$ representatives of parts in $Q''$,
that span a large number of $F$-copies for some $F \in \mathcal{F}$.
However, in the infinite case, such a definition of $d_{\mathcal{F}}$ cannot work.
Instead, we take $d_{\mathcal{F}}(m,t)$ to be a parameter that depends on the family $\mathcal{F}$, the size of the alphabet $|\Sigma|$ and the integers $m,t$ (where $m = |I|$, $mt = |Q'|$).
%rather than depending only on $\Sigma$ as in the finite case.
It is then shown that with this choice of $d_{\mathcal{F}}$, the proof follows similarly to the finite case, with Lemmas \ref{lem:representing_ramsey} and \ref{lem:has_good_copy}
being replaced with similar lemmas that are suitable for the infinite case (Lemmas \ref{lem:representing_ramsey_infinite} and \ref{lem:has_good_copy_infinite} below, respectively).

\subsection{Embeddability}
\begin{definition}[Embeddability]
\label{def:embeddability}
For a finite alphabet $\Sigma$, integers $m, t > 0$, $\Gamma(\Sigma, t)$-colored graph with loops $H = ([m], c_H)$ and $\Sigma$-colored graph $F = ([n_F], c_F)$,
we say that $F$ is \emph{embeddable} in $H$ if there exists a mapping $h:[n_F] \to V_H$ with the following properties.
\begin{itemize}
\item $h$ is weakly order-preserving: $h(1) \leq \ldots \leq h(n_F)$.
\item There exist integers $s_1, \ldots, s_{n_F} \in [t]$ so that $c_F(ii') \in c_H(h(i), h(i'))[s_i, s_{i'}]$ for any $i < i' \in [n_F]$.
\end{itemize}
A family $\mathcal{F}$ of $\Sigma$-colored graphs is \emph{embeddable} in $H$ if some $F \in \mathcal{F}$ is embeddable in $H$.
\end{definition}

The next lemma states that the desired $d_{\mathcal{F}}$ is indeed well-defined. It is similar in spirit to
the ideas of Alon and Shapira \cite{AlonShapira2008} (see Section 4 there). 
\begin{lemma}
\label{lem:choosing_num_representatives}
Fix a finite alphabet $\Sigma$. For any (finite or infinite) family $\mathcal{F}$ of $\Sigma$-ordered graphs and integers $m,t > 0$,
there exists $d_{\mathcal{F}} = d_{\mathcal{F}}^{(\ref{lem:choosing_num_representatives})}(m,t)$ with the following property.
If $H = ([m], c_H)$ is a $\Gamma(\Sigma, t)$-colored graph with loops, 
and if $\mathcal{F}$ is embeddable in $H$, then there is a graph $F \in \mathcal{F}$ which is embeddable in $H$ 
and has at most $d_{\mathcal{F}}^{(\ref{lem:choosing_num_representatives})}(m,t)$ vertices.
\end{lemma}
\begin{proof}
Let $\mathcal{H} = \mathcal{H}_{m,t}$ denote the set of all $\Gamma(\Sigma, t)$-colored graphs $H = ([m], c_H)$ with loops, such that $\mathcal{F}$ is embeddable in $H$.
Note that $|\mathcal{H}_{m,t}| \leq |\Gamma(\Sigma, t)|^{m^2} \leq 2^{|\Sigma| t^2 m^2}$. 
For any $H \in \mathcal{H}$ let $\mathcal{F}_H \subseteq \mathcal{F}$ denote the collection of all graphs in $\mathcal{F}$ that are embeddable in $H$.
Finally define 
\[
d_{\mathcal{F}}^{(\ref{lem:choosing_num_representatives})}(m,t) = \max_{H \in \mathcal{H}_{m,t}} \min_{F \in \mathcal{F}_H} |F|
\]
where $|F|$ denotes the number of vertices in $F$. Since $\mathcal{H}_{m,t}$ is finite, and since the set $\mathcal{F}_H$ is non-empty
for any $H \in \mathcal{H}$ (by definition of $\mathcal{H}$), 
the function $d_{\mathcal{F}}^{(\ref{lem:choosing_num_representatives})}(m,t)$ is well defined.
Now let $H$ be a graph as in the statement of the lemma and suppose that $\mathcal{F}$ is embeddable in $H$. 
Then $H \in \mathcal{H}_{m,t}$, so there exists $F \in \mathcal{F}_H$ of size at most $d_{\mathcal{F}}^{(\ref{lem:choosing_num_representatives})}(m,t)$.
\end{proof}

\subsection{Adapting the proof for infinite families}
For what follows, 
a \emph{nicely colored $(m,t)$-subgraph} is defined exactly like a nicely colored subgraph (see Definition \ref{def:nicely_colored_subgraph}), except that each set $D_j$ is of size $d_{\mathcal{F}}^{(\ref{lem:choosing_num_representatives})}(m,t)$. %Also define $r'(m, t) = R_{\ref{coro:ramsey_app}}(2^{|\Sigma| t^2} , m, d_{\mathcal{F}}^{(\ref{lem:choosing_num_representatives})}(m, t))$. Eventually we will pick $b = r'(m,t)$ ($r'$ play the role of the function $r$). 
The following is a variant of Lemma \ref{lem:representing_ramsey} for the infinite case.
\begin{lemma}
\label{lem:representing_ramsey_infinite}
For any two integers $m,t > 0$ there exists $R = R_{\ref{lem:representing_ramsey_infinite}}(m, t)$ satisfying the following: If $b \geq R_{\ref{lem:representing_ramsey_infinite}}(m, t)$, 
then for any tuple $W = (W_{111}, \ldots, W_{mbt})$ that represents $Q''$ and any
$\eta > 0$ there exists a nicely colored $(m,t)$-subgraph $D = D_{\ref{lem:representing_ramsey_infinite}}(W, \eta)$ of $H_W^{\eta}$.
Moreover, if $H_W^{\eta}$ is $\rho$-desirable for some $\eta < \rho < 1 / |\Sigma|$, then the number of $\rho$-undesirable edges in $D$ is at most $2 \rho \binom{m}{2} (d_{\mathcal{F}}^{(\ref{lem:choosing_num_representatives})}(m,t))^2$.
\end{lemma}
The proof of Lemma \ref{lem:representing_ramsey_infinite} is essentially identical to that of Lemma \ref{lem:representing_ramsey}, with any occurrence  of $d_\mathcal{F}$ replaced by $d_{\mathcal{F}}^{(\ref{lem:choosing_num_representatives})}(m,t)$. In particular we take
$R_{\ref{lem:representing_ramsey_infinite}}(m, t) = R_{\ref{coro:ramsey_app}}(2^{|\Sigma| t^2} , m, d_{\mathcal{F}}^{(\ref{lem:choosing_num_representatives})}(m, t))$. 

Next we state the variant of Lemma \ref{lem:easy_to_clean} for the infinite case. 
The proof is essentially identical.
\begin{lemma}
	\label{lem:easy_to_clean_infinite}
	Suppose that $D$ is a nicely colored $(m,t)$-subgraph
	of some $H_W^{\eta}$ with $W$ representing $Q''$ and $0 < \eta < \rho$, such that at most $2 \rho \binom{m}{2} (d_{\mathcal{F}}^{(\ref{lem:choosing_num_representatives})}(m, t))^2$ edges of $D$ are $\rho$-undesirable.
	Then $G' = G'(G, D)$ is $(7 |\Sigma| \rho + 2/m)$-close to $G$, where $m = |I|$.
\end{lemma}

The next lemma is the variant of Lemma \ref{lem:has_good_copy} that we use in the infinite case.
In contrast to the previous two lemmas, here the proof is slightly modified, and makes use of Lemma \ref{lem:choosing_num_representatives}.
\begin{lemma}
\label{lem:has_good_copy_infinite}
Let $W = (W_{111}, \ldots, W_{mbt})$ be a tuple that represents $Q''$ and let $\eta > 0$. 
If $D$ is a nicely colored $(m,t)$-subgraph of $H_W^{\eta}$ and $G'(G, D)$ contains a copy of a graph from $\mathcal{F}$,
then there exist $F = ([n_F], c_F) \in \mathcal{F}$, where $n_F \leq d_{\mathcal{F}}^{(\ref{lem:choosing_num_representatives})}(m,t)$, and sets 
$W_{j_1 r_1 s_1}, \ldots, W_{j_{n_F} r_{n_F} s_{n_F}} \in W$, with the following properties.
\begin{itemize}
\item For any $i \in [n_F-1]$, either $j_{i+1} > j_i$, or $j_{i+1} = j_i$ and $r_{i+1} > r_i$.
\item For any $i < i' \in [n_F]$ it holds that $d_{c_F(ii')}(W_{j_i r_i s_i}, W_{j_{i'} r_{i'} s_{i'}}) \geq \eta$.
\end{itemize}
\end{lemma}

\begin{proof}
Consider the $\Gamma$-colored graph with loops $D' = ([m], c_{D'})$: For any $j \leq j'$, $c_{D'}(jj') = C_{jj'}^{(D)}$.
Suppose that $G'(G,D) = (V, c')$ contains a copy of $A = ([n_A], c_A) \in \mathcal{F}$,
whose vertices in $V$ are $v_1 < \ldots < v_{n_A}$. 
For any $i \in [n_A]$, let $j_i \in [m], s_i \in [t]$ be the indices for which 
$v_i \in I_{j_i s_i}$. Then for any $i < i' \in [n_A]$ we have $c_A(ii') = c'(v_i v_i') \in C_{j_i j_{i'}}^{(D)}[s_i, s_{i'}] = c_{D'}(j_i j_{i'})[s_i, s_{i'}]$,
and so $A$ is embeddable in $D'$ (by the mapping  $i \mapsto j_i$).
By Lemma \ref{lem:choosing_num_representatives}, there exists $F = ([n_F], c_F)\in \mathcal{F}$ which is embeddable in $D'$, where $n_F \leq d_{\mathcal{F}}^{(\ref{lem:choosing_num_representatives})}(m,t)$. Let $h:[n_F] \to D'$ denote a mapping  that satisfies the 
conditions of Definition \ref{def:embeddability} and let $s'_1, \ldots, s'_{n_F} \in [t]$ be the indices satisfying 
$c_F(ii') \in c_{D'}(h(i), h(i'))[s'_{i}, s'_{i'}]$ for any $i < i' \in [n_F]$. 

For any $ i \in [n_F]$ denote the vertices of $D$ inside $I_{h(i)}$ by $I_{h(i) r'_{i1}}, \ldots, I_{h(i) r'_{id_{\mathcal{F}}(m,t)}}$, where $r'_{i1} < \ldots < r'_{i d_{\mathcal{F}}(m,t)} \in [b]$ for any $ i \in [n_F]$.
The sets $W_{h(1) r'_{11} s'_{1}}, \ldots, W_{h(n_F) r'_{n_F n_F} s'_{n_F}}$ satisfy the desired conditions: 
They exist, since $n_F \leq d_{\mathcal{F}}^{(\ref{lem:choosing_num_representatives})}(m,t)$,
the first condition holds since $h$ is order-preserving, and the second condition holds since
$c_F(ii') \in c_{D'}(h(i), h(i'))[s'_{i}, s'_{i'}] = C_{h(i)h(i')}^{(D)}[s'_i, s'_{i'}]$.
\end{proof}

\begin{proof}[Proof of Theorem \ref{thm:infinite_removal_lemma}]
The proof goes along the same lines as the proof of Theorem \ref{thm:finite_removal_lemma}, but 
any occurrence of $d_{\mathcal{F}}$ in the proof of Theorem \ref{thm:finite_removal_lemma} and the accompanying lemmas is replaced here by $d_{\mathcal{F}}^{(\ref{lem:choosing_num_representatives})}(m, t)$, including in the definitions of the functions $\beta, r$, and the term \emph{nicely colored subgraph} is replaced by \emph{nicely colored $(m, t)$-subgraph}.
More specifically, here are the exact changes needed with respect to the proof of Theorem \ref{thm:finite_removal_lemma}.
\begin{itemize}
	\item We take the functions $\beta = \beta_{\ref{lem:regularity_to_many_copies}}^{\rho/2}$ and $r = R_{\ref{lem:representing_ramsey_infinite}}$ (in the finite case we took $\beta$ as a suitable constant function and $r = R_{\ref{lem:representing_ramsey}}$). The function $g$ is defined
	as $g(l) = l r(l,l)$. Following the application of Lemma \ref{lem:mainpartscheme}, we have $b =  R_{\ref{lem:representing_ramsey_infinite}}(m, t)$.
	\item As in the the proof of Theorem \ref{thm:finite_removal_lemma}, there is a tuple $W$ which $(\alpha, \beta(mbt), \mu)$-represents $Q''$, and $H_{W}^{\rho/2}$ is $\rho$-desirable.
	By Lemma \ref{lem:representing_ramsey_infinite}, and by our new choice of $b$, there is a 
	nicely colored $(m,t)$-subgraph $D$ of $H_{W}^{\rho/2}$, with at 
	most $2\rho \binom{m}{2} \left(d_{\mathcal{F}}^{(\ref{lem:choosing_num_representatives})}(m, t)\right)^2$ $\rho$-undesirable edges.
	\item Lemma \ref{lem:easy_to_clean_infinite} implies that $G'$ contains a copy of a graph from $\mathcal{F}$. Now Lemma \ref{lem:has_good_copy_infinite} implies the existence of sets 
	$W_{j_1 r_1 s_1}, \ldots, W_{j_{n_F}, r_{n_F}, s_{n_F}} \in W$ with $n_F \leq d_{\mathcal{F}}^{(\ref{lem:choosing_num_representatives})}(m, t)$, that satisfy the conditions
	of the lemma for $\eta = \rho/2$. Since all pairs of sets from $W$ are $\beta_{\ref{lem:regularity_to_many_copies}}^{\rho/2}(mbt)$-regular,
	and since $mbt \geq b \geq d_{\mathcal{F}}^{(\ref{lem:choosing_num_representatives})}(m, t) \geq n_F$, these pairs are also $\beta_{\ref{lem:regularity_to_many_copies}}^{\rho/2}(n_F)$-regular. We apply Lemma \ref{lem:regularity_to_many_copies} to get that the number of $F$-copies in $G$ is at least 
	$\delta n^q$ for
	\begin{align*}
	q &= n_F \leq d_{\mathcal{F}}^{(\ref{lem:choosing_num_representatives})}(m, t) \leq d_{\mathcal{F}}^{(\ref{lem:choosing_num_representatives})}(S_{\ref{lem:mainpartscheme}}(\gamma,k,f,r), T_{\ref{lem:mainpartscheme}}(\gamma,k,f,r)), \\
	\delta &= \kappa_{\ref{lem:regularity_to_many_copies}}\left(\rho/2, n_F\right) \alpha^{n_F} 
	\geq \kappa_{\ref{lem:regularity_to_many_copies}}\left(\rho/2,  d_{\mathcal{F}}^{(\ref{lem:choosing_num_representatives})}(m, t)\right) \alpha^{d_{\mathcal{F}}^{(\ref{lem:choosing_num_representatives})}(m, t)} \\
	&\geq \kappa_{\ref{lem:regularity_to_many_copies}}(\rho/2, d_{\mathcal{F}}^{(\ref{lem:choosing_num_representatives})}(S_{\ref{lem:mainpartscheme}}(\gamma,k,f,r), T_{\ref{lem:mainpartscheme}}(\gamma,k,f,r))) \alpha^{d_{\mathcal{F}}^{(\ref{lem:choosing_num_representatives})}(S_{\ref{lem:mainpartscheme}}(\gamma,k,f,r), T_{\ref{lem:mainpartscheme}}(\gamma,k,f,r))}.
	\end{align*} 
\end{itemize}
Indeed, the above bounds for $q$ and $\delta$ depend only on $|\Sigma|, \epsilon, \mathcal{F}$, and not on $n$.
\end{proof}

\subsection{Adapting the proof for matrices}
\label{subsec:adapt_matrices}
Finally we give a sketch of the proof of Theorem \ref{thm:infinite_matrix_removal_lemma} for square matrices. The proof is very similar to the graph case,
so we only describe why the proof for graphs also works here. Finally, we describe shortly how the proof can be adapted to the case of non-square matrices.

\begin{proof}[Proof sketch for Theorem \ref{thm:infinite_matrix_removal_lemma}]
Given a square matrix $M:U \times V \to \Sigma$ where $U, V$ are ordered, and a family $\mathcal{F}$ of forbidden submatrices, consider the $\Sigma'$-colored graph $G = (U \cup V, c)$ where $\Sigma' = \Sigma \cup \{\sigma_0\}$ for some $\sigma_0 \notin \Sigma$, and the union $U \cup V$ is ordered as follows: All elements of $V$ come after all elements of $U$, and the internal orders of $U$ and $V$ remain as before. The edge colors of $G$ satisfy $c(uv) = M(uv)$ for any $u \in U$ and $v \in V$, and 
$c(uv) = \sigma_0$ otherwise. 

The proof now follows as in the graph case. It is important to note that while in the graph case one is allowed to change the color of
any edge, here we are not allowed to change the color of an edge from or to the color $\sigma_0$. However, the proof still works, by the following observations. 
\begin{itemize}
\item Since $|U| = |V|$, the number of vertices in $G$ is even, and so the interval partition $I$ obtained here ``respects the middle''. That is, each part $I_j$ of $I$ will be fully contained in $U$ or in $V$. Therefore, for every two parts $I_j, I_{j'}$ of $I$, either all edges in $I_j \times I_{j'}$ are colored by $\sigma_0$ or none of them is colored by $\sigma_0$.
\item It follows that the set of edges of the cleaned graph $G' = G'(G, D)$ that are colored by $\sigma_0$ is identical to that of $G$. In other words, to generate the cleaned graph we do not modify edge colors to or from $\sigma_0$. Since $G$ is made $\mathcal{F}$-free only by modifying colors between $U$ and $V$ to other colors in $\Sigma$, one needs to modify at least $\epsilon |U| |V|$ edge colors, so the proof follows without changing the main arguments.
\end{itemize}
\end{proof}

The above proof works for square matrices, but it can be adapted to general $m \times n$ matrices: If $m = \Theta(n)$, then the condition on $I$ needed is slightly different than respecting the middle, but this only slightly changes the structure of the equipartitions that we obtain via Lemma \ref{lem:mainpartscheme}, without significantly affecting the proof.
The proof can also be formulated for matrices with, say, $m = o(n)$ and $m = \omega(1)$, 
but then Lemma \ref{lem:mainpartscheme} needs to be especially adapted to accommodate the two ``types'' of vertices (row and column). Essentially we will have two interval equipartitions, one of the row vertices and one of the column vertices, along with their corresponding refinements.
Finally, the case where $m = \Theta(1)$ is essentially the case of testing one-dimensional strings; strings can be handled as per the discussion in the end of Subsection \ref{subsec:related_work}.

It is important to note that one cannot use Theorem \ref{thm:infinite_removal_lemma} as a black box to prove Theorem \ref{thm:infinite_matrix_removal_lemma}, as the distance of the graph $G$ to $\mathcal{F}$-freeness might (potentially) be significantly smaller than $\epsilon$, considering that the set of $\sigma_0$-colored edges in the $\mathcal{F}$-free graph that is closest to $G$ might differ from the set of $\sigma_0$-colored edges in $G$.

\end{document}